\documentclass[twocolumn]{autart}
\usepackage{amsmath,amssymb,mathtools}
\usepackage{graphicx,epstopdf,framed}
\usepackage{euscript,paralist,xcolor}

\theoremstyle{remark}
\newtheorem{theorem}{Theorem}[section]

\newtheorem{proposition}[theorem]{Proposition}
\newtheorem{definition}[theorem]{Definition}

\newtheorem{remark}[theorem]{Remark}

\numberwithin{equation}{section}

\newcommand{\BE}{\begin{equation}}
\newcommand{\rfb}[1]{\mbox{\rm
   (\ref{#1})}\ifx\undefined\stillediting\else:\fbox{$#1$}\fi}

\newfont{\roma}{cmr10 scaled 1200}
\renewcommand{\cline}{{\mathbb C}}

\newcommand{\nline}  {{\mathbb N}}
\newcommand{\rline}  {{\mathbb R}}

\newcommand{\tline}  {{\mathbb T}}

\newcommand{\dd}   {{\rm d}\hbox{\hskip 0.5pt}}
\newcommand{\Ascr} {{\mathcal A}}
\newcommand{\Bscr} {{\mathcal B}}

\newcommand{\Dscr} {{\mathcal D}}

\newcommand{\Lscr} {{\mathcal L}}

\newcommand{\Wscr} {{\mathcal W}}

\newcommand{\mm}    {{\hbox{\hskip 0.5pt}}}
\newcommand{\m}     {{\hbox{\hskip 1pt}}}

\newcommand{\bluff} {{\hbox{\raise 15pt \hbox{\mm}}}}
\newcommand{\sbluff}{{\hbox{\raise  7pt \hbox{\mm}}}}

\newcommand{\FORALL} {{\hbox{$\hskip 11mm \forall \;$}}}

\renewcommand{\theequation}{{\arabic{section}.\arabic{equation}}}

\begin{document}
\begin{frontmatter}

\title{Motion planning and approximate controllability of a moving cantilever beam with a tip-mass\thanksref{footnoteinfo}}

\thanks[footnoteinfo]{This work was supported by the Anusandhan National Research Foundation, India, via the Core Research Grant CRG/2023/004880 and A. Batra is supported by the Prime Minister's Research Fellowship via the grant RSPMRF0262. This paper was not presented at any IFAC meeting. } 

\author[IITB]{Soham Chatterjee}\ead{soham.chatterjee@iitb.ac.in},
\author[IITB]{Aman Batra}\ead{aman.batra@iitb.ac.in},
\author[IITB]{Vivek Natarajan}\ead{vivek.natarajan@iitb.ac.in}

\address[IITB]{Centre for Systems and Control, Indian Institute of Technology Bombay, Mumbai 400076, India}

\begin{keyword}
Approximate controllability, coupled PDE-ODE model, flatness, motion planning, non-uniform beam.
\end{keyword}

\begin{abstract}
Consider a non-uniform Euler-Bernoulli beam with a tip-mass at one end and a cantilever joint at the other end. The cantilever joint is not fixed and can itself be moved along an axis perpendicular to the beam. The position of the cantilever joint is the control input to the beam. The dynamics of the beam is governed by a coupled PDE-ODE model with boundary input. On a natural state-space, there exists a unique state trajectory for this beam model for every initial state and each twice continuously differentiable control input which is compatible with the initial state. In this paper, we study the motion planning problem of transferring the beam model from an initial state to a final state over a prescribed time-interval and then employ the results obtained to establish the approximate controllability of this model. We address these problems by
extending and applying the generating functions approach to flatness-based control
to the beam model. We prove that the transfer described above is feasible if the initial and final states belong to a certain set, which also contains the steady-states of the beam model. We then establish that this set contains all the eigenfunctions of the beam model, which form a Riesz basis for the state-space, and thereby conclude the approximate controllability of the beam model over all time intervals. We illustrate our theoretical results on motion planning using simulations and experiments. \vspace{-2mm}
\end{abstract}
\end{frontmatter}

\section{Introduction} \label{intro}

Consider the following model of a non-uniform moving cantilever Euler-Bernoulli beam with a tip-mass at one end and a cantilever joint at the other end: \vspace{-2mm}
\begin{align}
 & \rho(x) w_{tt}(x,t)+(EI w_{xx})_{xx}(x,t) = 0, \label{eq:beam1}\\[0.5ex]
 & \hspace{2mm} m w_{tt}(0,t)+(EI w_{xx})_x(0,t)=0, \label{eq:beam2}\\[0.5ex]
 & \hspace{2mm} J w_{xtt}(0,t) - EI(0)w_{xx}(0,t) = 0, \label{eq:beam3} \\[0.5ex]
 & \hspace{2mm} w(L,t) = f(t), \qquad w_x(L,t)=0. \label{eq:beam4}\\[-4.5ex] \nonumber
\end{align}
Here $L>0$ is the length of the beam, $w(x,t)$ is the displacement of the beam at the location $x\in [0,L]$ and time $t\in(0,\infty)$, $\rho(x)$ and $EI(x)$ are the mass per unit length and flexural rigidity, respectively, of the beam at $x\in[0,L]$, and $m>0$ and $J>0$ are the mass and moment of inertia, respectively, of the tip-mass located at the $x=0$ end of the beam. We suppose that $\rho \in C^4[0,L]$, $EI \in C^4[0,L]$ and they are strictly positive, i.e $\inf_{x\in[0,L]}\rho(x)>0$ and $\inf_{x\in[0,L]}EI(x)>0$. The displacement of the cantilever joint at the $x=L$ end of the beam is determined by the scalar control input $f$. The coupled PDE-ODE model \eqref{eq:beam1}-\eqref{eq:beam4} governs the dynamics of engineering systems which have a moving cantilever beam such as single-axis flexible cartesian robots. Figure 1 shows an experimental setup consisting of a non-uniform moving cantilever beam with a tip-mass. We remark that our model \eqref{eq:beam1}-\eqref{eq:beam4} resembles the SCOLE model that has been widely studied in the literature. Indeed, by fixing the cantilever joint at $x=L$ (i.e. taking $f=0$ in \eqref{eq:beam4}) and adding a force input $v_1$ and a torque input $v_2$ on the tip-mass at $x=0$ (i.e. letting the right-sides of  \eqref{eq:beam2} and \eqref{eq:beam3} be $v_1$ and $v_2$, respectively) we get the SCOLE model. \vspace{-1mm}

There exists a natural state-space $Z$ for the beam model \eqref{eq:beam1}-\eqref{eq:beam4} in which it has a unique state trajectory for every initial state $z_0$  and each twice continuously differentiable control input $f$ which is compatible with $z_0$, see Section \ref{sec2}. In this paper, we study the motion planning problem of  transferring \eqref{eq:beam1}-\eqref{eq:beam4} from an initial state $z_0$ to a final state $z_T$ over a prescribed time-interval $[0,T]$ using an appropriate control input $f$. We prove that the desired transfer is possible if the initial and final states belong to a certain set, which contains the steady-states of the beam model and is independent of $T$. We then establish that this set is dense in the state-space which implies that the beam model \eqref{eq:beam1}-\eqref{eq:beam4} is approximately controllable over all time intervals. \vspace{-5mm}

$$\includegraphics[width=0.45\textwidth]{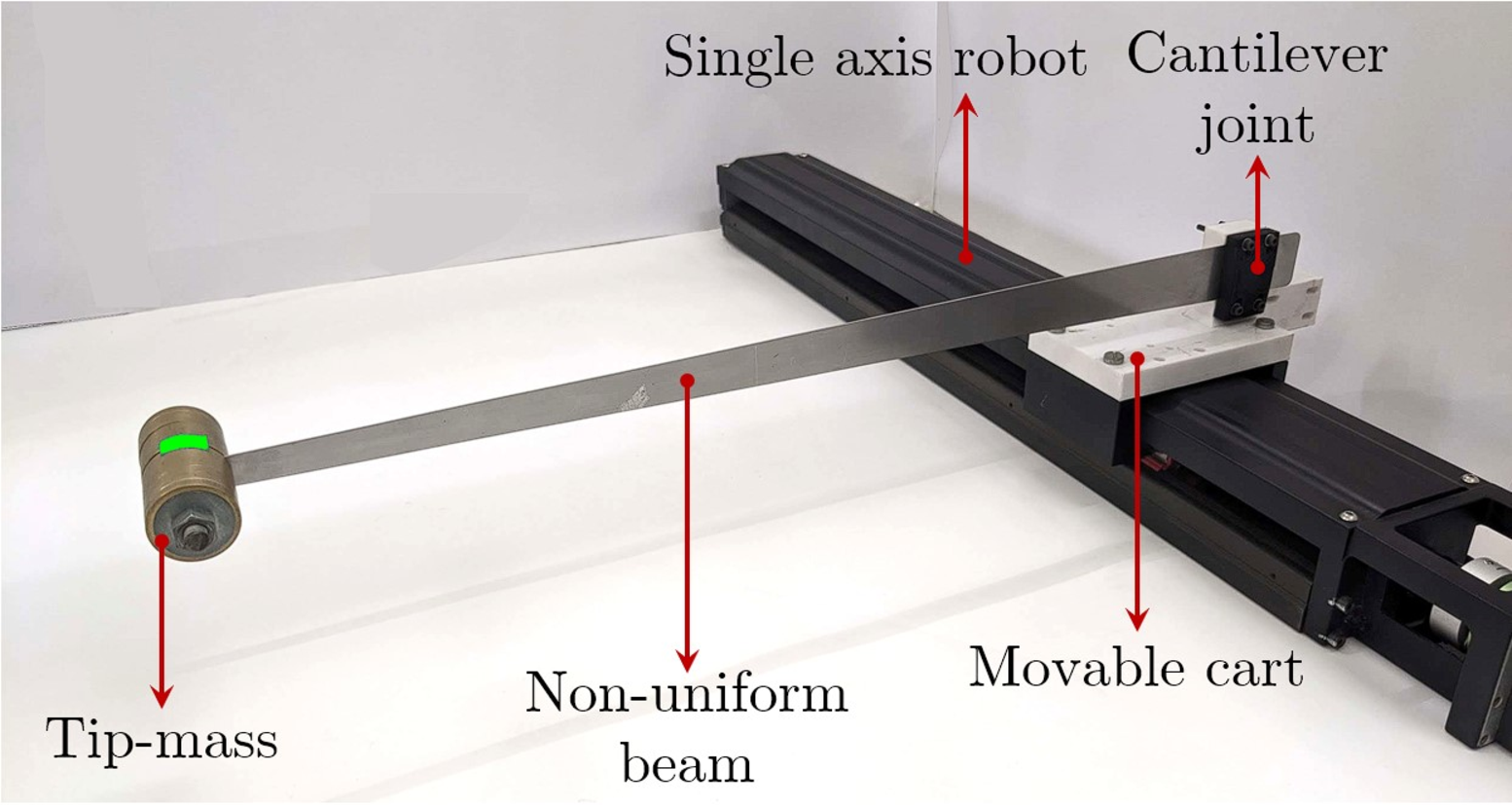} $$
\centerline{ \parbox{3.2in}{\small
Figure 1. Experimental setup of a non-uniform moving cantilever Euler-Bernoulli beam with a tip-mass, see Section \ref{sec6} for details. \vspace{2mm}
}}


The controllability of the SCOLE model, which resembles our beam  model \eqref{eq:beam1}-\eqref{eq:beam4} as discussed above, has been investigated comprehensively in the literature. It is well-known that the SCOLE model is not exactly controllable on the natural energy state-space over any time interval $[0,T]$ using $L^2$ control inputs, see for instance \cite{Guo:2002}. However, exact controllability can be recovered if one admits singular control inputs or considers smaller state-spaces, 
see \cite{Guo:2002}, \cite{GuIv:2005}, \cite{LiMa:1988}, \cite{Rao:2001}, \cite{ZhWe:2010}. In \cite{Rao:2001}, it has been shown via the Hilbert Uniqueness Method that the SCOLE model is exactly controllable on the energy state-space provided the torque input is allowed to belong to the dual of $H^1(0,T)$. For certain smaller and more regular state-spaces, it has been established in \cite{LiMa:1988} using the theory of semi-infinite beams, in \cite{Guo:2002} using spectral analysis and the Russell's exact controllability via stabilzability principle, in \cite{GuIv:2005} via the moment method and in \cite{ZhWe:2010} using an operator-theoretic framework for coupled linear systems that the SCOLE model is exactly controllable using
force and/or torque inputs in $L^2(0,T)$. These smaller state-spaces are dense subspaces of the energy state-space, which implies the approximate controllability of the SCOLE model on the energy state-space using $L^2$ inputs, see also \cite[Proposition 4.1]{Guo:2002}. In contrast to the SCOLE model, to the best of our knowledge, the controllability of the moving cantilever beam model \eqref{eq:beam1}-\eqref{eq:beam4} has not been studied in the literature. In this work, by employing the flatness technique, we prove that the beam model \eqref{eq:beam1}-\eqref{eq:beam4} is approximately controllable on its natural state-space $Z$ using $L^2$ displacement inputs. 

In the flatness technique for motion planning of infinite-dimensional control systems, the state and the input of the system are expressed as an infinite linear combination of a flat output and its time derivatives. Then, based on the desired motion, an appropriate trajectory is selected for the flat output using which the input necessary to execute the motion is computed. The flatness technique has been used in the literature to address motion planning problems for Euler-Bernoulli beams. Using the transform approach (Laplace transform or Mikusi\'nski's operational calculus) to flatness, motion planning problems have been solved for Euler-Bernoulli beams with fixed cantilever joints in \cite{MeThKu:2008}, rotating cantilever joints in \cite{AoFlMoRoRu:1997}, \cite{BaLy:2008} and \cite{LyWa:2004} and translating cantilever joints in \cite{BaUlRu:2011}. In \cite{RuWo:2008}, it has been shown that the transform approach to flatness can be used to address motion planning problems for a class of linear PDEs, which includes beam models. The power series approach and Riesz spectral approach to flatness have been used in \cite{KoGeSc:2022} and \cite{MeScKu:2010}, respectively, to solve motion planning problems for Euler-Bernoulli beams with fixed cantilever joints. Recently, we used the semi-discretization approach to flatness in \cite{ChNa:2020} and \cite{ChBaNa:2025} to solve motion planning problems for Euler-Bernoulli beams with hinged and sliding cantilever joints, respectively. While the above works address 
motion planning problems for Euler-Bernoulli beams using the flatness technique, they do not examine the controllability properties of these beams. \vspace{-1mm}

Another approach to flatness based motion planning of 1D PDEs is the generating functions approach. In this approach the input and solution of the PDE are expressed in terms of certain generating functions which are obtained by solving a sequence of initial value ODEs recursively, see \cite{LaMa:2000}. This approach has been used to establish the null controllability of 1D parabolic PDEs in \cite{MaRoRo:2014}, \cite{MaRoRo:2016}, 1D Schr\"odinger equations in \cite{MaRoRo:2018} and the linearized Korteweg-de Vries equations \cite{MaRiRoRo:2019}. In \cite{MaRoRo:2014}, the null controllability of nD heat equations on cylindrical domains has also been established. In the present work, we extend this generating functions approach to solve the problem of transferring the beam model \eqref{eq:beam1}-\eqref{eq:beam4} from an initial state $z_0$ to a final state $z_T$ over a prescribed time-interval $[0,T]$. We prove that this transfer is possible provided $z_0$ and $z_T$  belong to a certain time-independent set, which contains the steady-states of the beam model. 
We establish that the eigenfunctions of \eqref{eq:beam1}-\eqref{eq:beam4} can be expressed in terms of the generating functions and conclude that the above set also contains the eigenfunctions. Since the eigenfunctions form a Riesz basis for the state-space, the approximate controllability of the beam model over any time interval follows. \vspace{-2mm}

The rest of the paper is organized as follows: In Section \ref{sec2} we  establish the well-posedness of the Euler-Bernoulli beam model \eqref{eq:beam1}-\eqref{eq:beam4}. We define the generating functions for the beam model in Section \ref{sec3} and derive some estimates for them. Section \ref{sec4} contains our solution to the motion planning problem and Section \ref{sec5} contains our proof of approximate controllability. We present our numerical and experimental results which illustrate our solution to the motion planning problem in Section \ref{sec6}.

{\bf Notations}. \\[1ex]
Let $H^k(0,L)$ denote the usual $k^{\rm th}$-order Sobolev space of complex-valued functions on the interval $(0,L)$. We denote the space of continuous and $k$-times continuously differentiable functions from an interval $[a,b]$ to a Hilbert space $X$ by $C([a,b];X)$ and $C^k([a,b];X)$, respectively, and they are both Banach spaces with the usual norm. Let $C^\infty([a,b];X)$ denote the set of all functions which belong to $C^k([a,b];X)$ for every $k\geq 0$. We write $C^k[a,b]$  instead of $C^k([a,b];\cline)$ and $C^\infty[a,b]$ instead of $C^\infty([a,b];\cline)$. For a function  $y\in C^\infty([a,b];X)$, let $y^{(m)}$ denote its $m^\mathrm{th}$-derivative. A function $y\in C^\infty[0,T]$ is said to be a Gevrey function of order $s>0$ if it satisfies the estimate $\sup_{t\in[0,T]}|y^{(m)}(t)|\leq D^{m+1}(m!)^{s}$ for all $m\in \nline$ and some constant $D>0$. The set of all functions which satisfy these estimates is denoted by $G_s[0,T]$. We will use a \emph{bar} to denote complex conjugate.

\section{Well-posedness} \label{sec2}

A natural state space for the coupled PDE-ODE beam model \eqref{eq:beam1}-\eqref{eq:beam4} is the complex Hilbert space \vspace{-1mm}
$$ Z = \{[u \ \  v \ \ \alpha \ \ \beta] \in H^2(0,L)\times L^2(0,L)\times \cline\times \cline \big| u_x(L)=0 \} \vspace{-1mm}$$
in which the inner product of any two vectors $z_1=[u_1 \ v_1 \ \alpha_1 \ \beta_1]$ and $z_2=[u_2 \ v_2 \ \alpha_2 \ \beta_2]$ is given by \vspace{-2mm}
\begin{align*}
 \langle z_1, z_2 \rangle_Z &= \int_0^L EI(x) u_{1, xx}(x) \overline{u}_{2, xx}(x) \dd x \\
 &\quad + \int_0^L u_1(x) \overline u_2(x) \dd x  + \int_0^L \rho(x) v_1(x) \overline v_2(x) \dd x \\
 &\quad + m\alpha_1 \overline \alpha_2 + J \beta_1 \overline \beta_2. \\[-4.5ex]
\end{align*}
Note that each physically meaningful configuration of the Euler-Bernoulli beam considered in this work, see Figure 1, with a smooth beam shape and a translated cantilever position, can be represented by a vector in the state-space $Z$. In this section we will establish the existence and uniqueness of weak solutions to the beam model \eqref{eq:beam1}-\eqref{eq:beam4} on the state-space $Z$ under a compatibility assumption on the control input $f$. We will show that these solutions depend continuously on the initial state and the control input, provided the compatibility assumption is not violated.

\begin{definition}
A function $f\in C[0,T]$ is said to be compatible for a vector $[u \ v \ \alpha \ \beta]\in Z$ if $f(0)=u(L)$.
\end{definition}

We now introduce a natural notion of weak solutions for the beam model \eqref{eq:beam1}-\eqref{eq:beam4} when its control input is compatible for its initial state.

\begin{definition}\label{def:weaksol}
Let $T>0$, $z_0 = [u_0 \ v_0 \ \alpha_0 \ \beta_0]\in Z$ and $f\in C[0,T]$ satisfying $f(0)=u_0(L)$ be given. A function $z\in C([0,T];Z)$ given as \vspace{-1mm}
\begin{equation}\label{eq:wdetz}
 z(t)=[w(\cdot,t) \ \ w_t(\cdot,t) \ \ w_t(0,t) \ \ w_{xt}(0,t)],
\end{equation}
where $w\in C([0,T];H^2(0,L))\cap C^1([0,T];L^2(0,L))$ with \hfill
$w(0,\cdot),w_x(0,\cdot)\in C^1[0,T]$, is a weak solution of \eqref{eq:beam1}-\eqref{eq:beam4} on the time interval $[0,T]$ for the initial state $z_0$ and control input $f$ if $z(0)=z_0$ and $w(L,t)=f(t)$ for all $t\in[0,T]$ and the following equation holds for all $\varphi\in H^2(0,L)$ satisfying $\varphi(L)=\varphi_x(L)=0$ and all $t\in[0,T]$\vspace{-1mm}:
\begin{align}
 &\int_0^L \rho(x)\bigl[w_t(x,t) - v_0(x)\bigr]\varphi(x)\dd x  \nonumber\\
 & \quad + m \bigl[w_t(0,t)-\alpha_0\bigr]\varphi(0) + J\bigl[w_{xt}(0,t)-\beta_0\bigr]\varphi_x(0) \nonumber\\
 &\qquad = - \int_0^t \int_0^L EI(x) w_{xx}(x,\tau) \varphi_{xx}(x)\dd x \dd \tau. \label{eq:weaksolution} \\[-6.5ex]\nonumber
\end{align}
\end{definition}

Note that the weak form of \eqref{eq:beam1}-\eqref{eq:beam4} shown in \eqref{eq:weaksolution} is obtained by multiplying \eqref{eq:beam1} with $\varphi$ and then integrating it by parts formally by taking into account the boundary conditions \eqref{eq:beam2}-\eqref{eq:beam4} and the initial condition $w_t(\cdot,0)=v_0$. We remark that our notion of weak solution includes a compatibility condition which couples the input to the initial state. This is to be expected since any solution $z(t)=[w(\cdot,t) \ w_t(\cdot,t) \ w_t(0,t) \ w_{xt}(0,t)]$ for the beam model \eqref{eq:beam1}-\eqref{eq:beam4} in $Z$, defined using any notion of solution, uniquely determines $w(L,t)$ (which is the control input).

In the next proposition, we show that \eqref{eq:beam1}-\eqref{eq:beam4} is well-posed in $Z$ subject to the compatibility assumption, i.e. we show that when the control input $f\in C^2[0,T]$ is compatible for the initial state, there exists a unique weak solution for \eqref{eq:beam1}-\eqref{eq:beam4} which depends continuously on $f$ and the initial state.

\begin{proposition}\label{pr:wellposed}
Let $T>0$, an input $f\in C^2[0,T]$ and an initial state $z_0 = [u_0 \ v_0 \ \alpha_0 \ \beta_0]\in Z$ be given. Suppose that $f$ is compatible for $z_0$, i.e. $f(0)=u_0(L)$. Then there exists a unique weak solution $z\in C([0,T];Z)$ of \eqref{eq:beam1}-\eqref{eq:beam4} on the time interval $[0,T]$ for the initial state $z_0$ and control input $f$. Furthermore, there exists a $M_T>0$ independent of $z_0$ and $f$ such that \vspace{-1mm}
\begin{equation}\label{eq:est_wellposed}
  \sup_{t\in [0,T]}\|z(t)\|_{Z} \leq M_T \Big(\|z_0\|_Z + \|f\|_{C^2[0,T]}\Big). \vspace{-4mm}
\end{equation}
\end{proposition}

\begin{pf}
Fix $\nu\in C^\infty[0,L]$ such that $\nu(0)=\nu_x(0)=\nu_{xx}(0)= \nu_{xxx}(0)=\nu_x(L)=0$ and $\nu(L)=1$. Replacing $w(x,t)$ with $\tilde w(x,t)+\nu(x)f(t)$ formally in \eqref{eq:beam1}-\eqref{eq:beam4} we obtain the following coupled PDE-ODE system: For $x\in (0,L)$ and $t>0$,
\begin{align}
  &\rho(x)\tilde w_{tt}(x,t)+ (EI(x) \tilde w_{xx})_{xx}(x,t)\nonumber\\[0.5ex]
  &\hspace{2mm} + (EI(x) \nu_{xx})_{xx}(x)f(t)+\rho(x)\nu(x)\ddot{f}(t)=0, \label{eq:wbeam1}\\[0.5ex]
  &m \tilde w_{tt}(0,t)+(EI \tilde w_{xx})_x(0,t)=0, \label{eq:wbeam2} \\[0.5ex]
  &J\tilde w_{xtt}(0,t) - EI(0)\tilde w_{xx}(0,t) = 0, \label{eq:wbeam3} \\[0.5ex]
  &\hspace{5mm}\tilde w(L,t) =0, \quad  \tilde w_x(L,t)=0. \label{eq:wbeam4}
\end{align}
Consider the Hilbert space $\tilde Z = \{[u \ v \ \alpha \ \beta] \in Z \m\big|\m u(L)=0\}$ equipped with the following inner product: For $z_1=[u_1 \ v_1 \ \alpha_1 \ \beta_1] \in \tilde Z$ and $z_2=[u_2 \ v_2 \ \alpha_2 \ \beta_2]\in \tilde Z$, \vspace{-2mm}
\begin{align}
 &\langle z_1, z_2  \rangle_{\tilde Z} = \int_0^L EI(x)u_{1,xx}(x) \overline u_{2,xx}(x) \dd x \nonumber\\
 &\hspace{5mm} + \int_0^L \rho(x)v_1(x) \overline v_2(x)\dd x + m\alpha_1 \overline \alpha_2 + J\beta_1 \overline\beta_2. \label{eq:tZinner} \\[-7ex]\nonumber
\end{align}

Recall $z_0$ and $f$ from the statement of this proposition. Define $\tilde z_0 = [\tilde u_0 \ \tilde v_0 \ \tilde \alpha_0 \ \tilde \beta_0]$, where $\tilde u_0 = u_0 - \nu f(0)$, $\tilde v_0 = v_0 - \nu \dot f(0)$, $\tilde \alpha_0 = \alpha_0$ and $\tilde \beta_0 = \beta_0$. Then $\tilde z_0\in \tilde Z$.
We claim that there exists a unique function $\tilde z\in C([0,T];\tilde Z)$, given as \vspace{-1mm}
\begin{equation}\label{eq:tdwdetstdz}
 \tilde z(t)=[\tilde w(\cdot,t) \ \ \tilde w_t(\cdot,t) \ \ \tilde w_t(0,t) \ \ \tilde w_{xt}(0,t)] \vspace{-1mm}
\end{equation}
with $\tilde w \in C([0,T];H^2(0,L))\cap C^1([0,T];L^2(0,L))$ and $\tilde w(0,\cdot), \tilde w_x(0,\cdot)\in C^1[0,T]$, is such that $\tilde z(0)=\tilde z_0$ and $\tilde w(L,t)=\tilde w_x(L,t)=0$ for all $t\in [0,T]$ and the following equation holds for all $\varphi\in H^2(0,L)$ satisfying $\varphi(L)=\varphi_x(L)=0$ and all $t\in [0,T]$: \vspace{-2mm}
\begin{align}
 &\int_0^L \rho(x)\bigl[\tilde w_t(x,t) -  \tilde v_0(x) - \nu(x)( \dot f(t)-\dot f(0))\bigr]\varphi(x)\dd x  \nonumber\\
 &\hspace{5mm} + m \bigl[\tilde w_t(0,t)- \tilde \alpha_0\bigr]\varphi(0) + J\bigl[\tilde w_{xt}(0,t)- \tilde \beta_0\bigr]\varphi_x(0) \nonumber\\
 &= -\int_0^t\!\!\int_{0}^{L}\!\! EI(x) \bigl[\tilde w_{xx}(x,\tau) - \nu_{xx}(x) f(\tau) \bigr]\varphi_{xx}(x)\dd x \dd \tau. \label{eq:weaksol_wbeam}\\[-4.5ex]\nonumber
\end{align}
Furthermore, $\tilde z$ satisfies \vspace{-1mm}
\begin{equation}\label{eq:ztd_wellposed_est}
 \sup_{t\in [0,T]} \|\tilde z(t) \|_{\tilde Z} \leq \tilde M_T\Big( \|\tilde z_0\|_{\tilde Z} + \|f\|_{C^2[0,T]}\Big). \vspace{-2mm}
\end{equation}
Equation \eqref{eq:weaksol_wbeam} is the weak form of the coupled PDE-ODE system \eqref{eq:wbeam1}-\eqref{eq:wbeam4}. We have presented a proof of the above claim and the estimate \eqref{eq:ztd_wellposed_est} in \vspace{-1mm} Appendix A.

Let $\tilde w$ and $\tilde z$ be the unique functions in the claim above. Define $w=\tilde w+\nu f$ and then define $z$ via \eqref{eq:wdetz} so \vspace{-1mm} that
\begin{equation}\label{eq:xprs}
 z(t) = \tilde z(t)+ [\nu f(t) \ \ \nu \dot f(t) \ \ 0 \ \ 0] \qquad \forall\, t\in [0,T]. \vspace{-1mm}
\end{equation}
Since $\tilde w$ and $\tilde z$ satisfy all the properties mentioned below \eqref{eq:tdwdetstdz}, it is easy to check that $w$ and $z$ satisfy all the properties mentioned in Definition \ref{def:weaksol}, i.e. $z\in C([0,T];Z)$ and it is a weak solution of \eqref{eq:beam1}-\eqref{eq:beam4} on the time interval $[0,T]$ for the initial state $z_0$ and control input $f$. The uniqueness of $z$ follows from the uniqueness of $\tilde z$. \vspace{-1mm}

Finally, the estimate in \eqref{eq:est_wellposed} can be established as follows: Note that $\tilde Z \subset Z$, $\|z^0\|_{\tilde Z} \leq \|z^0\|_{Z}$ and  $\|z^0\|_{Z} \leq c \|z^0\|_{\tilde Z}$ for all $z^0\in \tilde Z $ and some $c>0$ and $\|[\nu \ \nu \ 0 \ 0]\|_{Z}\leq c_\nu$ for some $c_\nu>0$. Applying the triangle inequality to \eqref{eq:xprs}, which gives us $\|z(t)\|_{Z} \leq c\|\tilde z(t)\|_{\tilde Z} + c_{\nu}\|f\|_{C^2[0,T]}$, and then using \eqref{eq:ztd_wellposed_est} we get
$$\|z(t)\|_{Z} \leq  c\tilde M_T\|\tilde z_0\|_{Z} + (c \tilde M_T + c_{\nu}) \|f\|_{C^2[0,T]}.$$
Since $\tilde z_0 = z_0 - [\nu f(0) \ \nu \dot f(0) \ 0 \ 0]$ (see the definition of $\tilde z_0$ above \eqref{eq:tdwdetstdz}) it follows that $\|\tilde z_0\|_Z \leq \| z_0\|_Z + c_\nu \|f\|_{C^2[0,T]}$. Using this in the above estimate for $\|z(t)\|_Z$ yields $\|z(t)\|_{Z} \leq c \tilde M_T\| z_0\|_{Z} + (c\tilde M_T + c c_\nu \tilde M_T  + c_{\nu}) \|f\|_{C^2[0,T]},$
i.e. \eqref{eq:est_wellposed} holds with $M_T=c \tilde M_T + c c_\nu \tilde M_T + c_{\nu}$. \hfill$\blacksquare$
\end{pf}

\begin{remark} \label{rm:classical}
Suppose that $w\in C([0,T];H^4(0,L))\cap C^2([0,T];L^2(0,L))$ with $w(0,t)\in C^2[0,T]$ and $w_x(0,t)\in C^2[0,T]$ satisfies \eqref{eq:beam1}-\eqref{eq:beam4} for some $f\in C^2[0,T]$ and each $t\in [0,T]$. Let $v_0=w_t(\cdot,0)$, $\alpha_0=w_t(0,0)$ and $\beta_0=w_{xt}(0,0)$. Multiplying \eqref{eq:beam1} with $\varphi\in H^2(0,L)$ and then using integration by parts it is easy to see that \eqref{eq:weaksolution} holds for each $t\in[0,T]$ and each $\varphi\in H^2(0,L)$ satisfying $\varphi(L)=\varphi_x(L)=0$. Hence $z\in C([0,T];Z)$ determined by $w$ via \eqref{eq:wdetz} is a weak solution of the beam model \eqref{eq:beam1}-\eqref{eq:beam4} for the initial state $z_0=[w(\cdot,0) \ w_t(\cdot,0) \ w_t(0,0) \ w_{xt}(0,0)]$ and the control input $f$, see Definition \ref{def:weaksol}. The uniqueness of this weak solution follows from Proposition \ref{pr:wellposed}. \hfill$\square$
\end{remark}

\section{Generating functions } \label{sec3}

In this paper we solve a motion planning problem for the beam model \eqref{eq:beam1}-\eqref{eq:beam4} by using the generating functions approach to flatness. 
Accordingly, we consider functions $w\in C^\infty([0,T]; C^4[0,L])$ which satisfy \eqref{eq:beam1}-\eqref{eq:beam4} for some $f\in C^\infty[0,T]$ and can be expressed as \vspace{-1mm}
\begin{equation}\label{eq:formalsoln}
 w(x,t) = \sum_{k\geq 0} g_k(x)y_1^{(2k)}(t) + \sum_{k\geq 0} h_k(x)y_2^{(2k)}(t) \vspace{-1.5mm}
\end{equation}
for all $x\in[0,L]$ and $t\in [0,T]$. Here $g_k$ and $h_k$ belonging to $C^4[0,L]$ are certain generating functions and $y_1, y_2\in G_s[0,T]$ with $1<s<2$ are two outputs which satisfy a constraint, see \eqref{eq:2ndinput}. The expressions for the generating functions are presented later in this section. In Proposition \ref{pr:wisclasic} we show that, given $y_1, y_2\in G_s[0,T]$ which satisfy \eqref{eq:2ndinput}, the function $w$ determined by the series in \eqref{eq:formalsoln} indeed belongs to $C^\infty([0,T]; C^4[0,L])$ and satisfies \eqref{eq:beam1}-\eqref{eq:beam4} for some $f\in C^\infty[0,T]$. \vspace{-1mm}

We require the outputs $y_1$ and $y_2$ to satisfy \vspace{-1mm}
\begin{equation} \label{eq:flatoutputs}
 y_1(t) = w(0,t), \qquad y_2(t) = w_x(0,t). \vspace{-1mm}
\end{equation}
In \eqref{eq:beam1}-\eqref{eq:beam3} and \eqref{eq:flatoutputs}, we replace $w$ with the series in \eqref{eq:formalsoln} formally. Then, for each $k\geq0$, equating the coefficients of $y_1^{(2k)}$ and $y_2^{(2k)}$ on either sides of the resulting expressions yields a sequence of initial value ODEs for $g_k$ and $h_k$ which can be solved recursively. These ODEs are given below. 
\vspace{-1mm}

The sequence of fourth-order linear ODEs which must be solved recursively, on the interval $x\in[0,L]$, for obtaining the real-valued generating functions $g_k$ are as follows: $g_0$ is obtained by solving the ODE \vspace{-1mm}
\begin{align}
  &\qquad\quad(EI g_{0,xx})_{xx}(x) =0, \label{eq:g0ODE}\\
  &\qquad g_0(0) = 1, \qquad  g_{0,x}(0) = 0, \label{eq:g0IC1}\\
  &g_{0,xx}(0) = 0, \qquad (EI g_{0,xx})_x(0) = 0,\label{eq:g0IC2}
\end{align}
$g_1$ is obtained by solving the ODE \vspace{-1mm}
\begin{align}
  &\quad(EI g_{1,xx})_{xx}(x) + \rho(x)g_0(x)=0, \label{eq:g1ODE}\\
  &\qquad g_1(0) = 0, \qquad  g_{1,x}(0) = 0, \label{eq:g1IC1}\\
  &g_{1,xx}(0) = 0, \qquad (EI g_{1,xx})_x(0) = -m,\label{eq:g1IC2}
\end{align}
and $g_k$ for $k\geq 2$ is obtained by solving the ODE \vspace{-1mm}
\begin{align}
  &\quad(EI g_{k,xx})_{xx}(x) + \rho(x)g_{k-1}(x)=0, \label{eq:gkODE}\\
  &\qquad g_k(0) = 0, \qquad  g_{k,x}(0) = 0,\label{eq:gkIC1}\\
  &\ \ g_{k,xx}(0) = 0, \qquad (EI g_{k,xx})_x(0) = 0.\label{eq:gkIC2}
\end{align}
The sequence of fourth-order linear ODEs which must be solved recursively, on the interval $x\in[0,L]$, for obtaining the real-valued generating functions $h_k$ are as follows: $h_0$ is obtained by solving the ODE \vspace{-1mm}
\begin{align}
  &\quad\qquad(EI h_{0,xx})_{xx}(x) =0, \label{eq:h0ODE}\\
  &\qquad h_0(0) = 0, \qquad  h_{0,x}(0) = 1, \label{eq:h0IC1}\\
  &h_{0,xx}(0) = 0, \qquad (EI h_{0,xx})_x(0) = 0,\label{eq:h0IC2}
\end{align}
$h_1$ is obtained by solving the ODE \vspace{-1mm}
\begin{align}
  &\ \ \quad (EI h_{1,xx})_{xx}(x) + \rho(x) h_{0}(x)=0,  \label{eq:h1ODE}\\
  &\ \ \qquad h_1(0) = 0, \qquad  h_{1,x}(0) = 0, \label{eq:h1IC1}\\
  &EI(0)h_{1,xx}(0) = J, \qquad (EI h_{1,xx})_x(0) = 0, \label{eq:h1IC2}
\end{align}
and $h_k$ for $k\geq 2$ is obtained by solving the ODE \vspace{-1mm}
\begin{align}
  &\quad (EI h_{k,xx})_{xx}(x)\! + \rho(x)h_{k-1}(x)=0, \label{eq:hkODE}\\
  &\qquad h_k(0) = 0, \qquad  h_{k,x}(0) = 0, \label{eq:hkIC1}\\
  & h_{k,xx}(0) = 0, \qquad (EI h_{k,xx})_x(0) = 0.\label{eq:hkIC2}  \\[-4.5ex]\nonumber
\end{align}

In the following proposition we show that the generating functions $g_k$ and $h_k$ belong to $C^4[0,L]$ and derive some estimates for them. 
\begin{proposition}\label{pr:genfnest}
For each $k\geq1$ the generating functions $g_k$ and $h_k$ belong to $C^4[0,L]$ and there exist positive constants $R_1$ and $R_2$ independent of $k$ such that the following estimates hold for all $x\in [0,L]$: \vspace{-1mm}
\begin{align}
 &|g_k(x)|\leq \frac{R_1^k\, x^{4k-1}}{(4k-1)!}, \qquad   |g_{k,x}(x)| \leq \frac{R_1^k \,x^{4k-2}}{(4k-2)!}, \label{eq:estggx}\\[0.5ex]
 &|h_k(x)|\leq \frac{R_2^k\, x^{4k-2}}{(4k-2)!}, \qquad |h_{k,x}(x)| \leq \frac{R_2^k\, x^{4k-3}}{(4k-3)!}. \label{eq:esthhx} \\[-4.5ex]\nonumber
\end{align}
\end{proposition}

\begin{pf}
Solving \eqref{eq:g0ODE}-\eqref{eq:g0IC2} for $g_0$ we get \vspace{-1mm}
\begin{equation} \label{eq:g0}
 g_0(x)=1 \FORALL x\in[0,L].
\end{equation}
Solving \eqref{eq:g1ODE}-\eqref{eq:g1IC2} for $g_1$ we get \vspace{-1mm}
\begin{align}
  &g_1(x) = -\int_0^x \!\int_0^{s_1}\!\int_0^{s_2}\! \int_0^{s_3 }\! \frac{\rho(s_4)}{EI(s_2)}\dd s_4 \dd s_3 \dd s_2 \dd s_1 \nonumber\\
  &\hspace{15mm}-m\int_0^x \int_0^{s_1} \int_0^{s_2}\frac{1}{EI(s_2)}\dd s_3 \dd s_2 \dd s_1 , \label{eq:g1sol} \\[0.5ex]
  &g_{1,x}(x) = -\int_0^x \!\int_0^{s_1}\!\int_0^{s_2}\! \frac{\rho(s_3)}{EI(s_1)}\dd s_3 \dd s_2 \dd s_1 \nonumber\\
  &\hspace{20mm}-m\int_0^x \int_0^{s_1} \frac{1}{EI(s_1)} \dd s_2 \dd s_1 \label{eq:g1xsol} \\[-4.5ex]\nonumber
\end{align}
for each $x\in[0,L]$. Since $EI$ and $\rho$ are strictly positive functions belonging to $C^4[0,L]$, it follows from the above equations via successive integrations that $g_1 \in C^4[0,L]$ and the estimates in \eqref{eq:estggx} hold for $k=1$ with \vspace{-3mm}
\begin{equation}\label{eq:R1defn}
  R_1 = \Bigg(\frac{\max_{x\in [0,L]}\rho(x)+m}{\min_{x\in [0,L]} EI(x)} \Bigg) \max\{1,L\}. \vspace{-3mm}
\end{equation}

Solving \eqref{eq:gkODE}-\eqref{eq:gkIC2} for $g_k$ we get that \vspace{-1mm}
\begin{align}
 &\!g_k(x) = -\int_{0}^{x}\!\int_{0}^{s_1}\!\!\int_{0}^{s_2}\!\! \int_{0}^{s_3}\! \frac{\rho(s_4)g_{k-1}(s_4)}{EI(s_2)}\dd s_4 \dd s_3 \dd s_2 \dd s_1,  \label{eq:gksol} \\
 &\!g_{k,x}(x) = -\int_{0}^{x}\!\int_{0}^{s_1}\!\!\int_{0}^{s_2}\! \frac{\rho(s_3)g_{k-1}(s_3)}{EI(s_1)} \dd s_3 \dd s_2 \dd s_1 \label{eq:gkxsol} \\[-4.8ex]\nonumber
\end{align}
for each $k\geq2$ and $x\in[0,L]$. Since $g_1,1\big/EI, \rho \in C^4[0,L]$  and $g_k$ is obtained by integrating $g_{k-1}, 1\big/EI$ and $\rho$ for each $k\geq2$, see \eqref{eq:gksol}, we get that $g_k\in C^4[0,L]$ for all $k\geq1$. Suppose that the estimates in \eqref{eq:estggx} hold for some $k=n$ with $n\geq 1$. Taking $k=n+1$ in \eqref{eq:gksol} and \eqref{eq:gkxsol}, and then bounding the integrand on the right using $\max_{x \in [0,L]} \rho(x)$, $\min_{x \in [0,L]} EI(x)$ and the estimate for $|g_n(x)|$ obtained from \eqref{eq:estggx} with $k=n$, it follows via successive integrations 
that the estimates in \eqref{eq:estggx} hold for $k=n+1$. We have now shown that the estimates in \eqref{eq:estggx} hold for $k=1$ and that they hold for $k=n+1$ if they hold for $k=n$ for any $n\geq1$. So it follows via the principle of mathematical induction that the estimates in \eqref{eq:estggx} hold for all $k\geq1$. This completes the proof of the claims regarding the generating functions $g_k$.

Next we will prove the estimates in \eqref{eq:esthhx}. Solving \eqref{eq:h0ODE}-\eqref{eq:h0IC2} for $h_0$ we get \vspace{-1mm}
\begin{equation} \label{eq:h0}
 h_0(x)=x \FORALL x\in[0,L]. \vspace{-1mm}
\end{equation}
Solving \eqref{eq:h1ODE}-\eqref{eq:h1IC2} for $h_1$ we get \vspace{-1mm}
\begin{align}
 &\hspace{-3mm} h_1(x) = -\int_0^x \!\int_0^{s_1}\!\int_0^{s_2}\! \int_0^{s_3 }\! \frac{s_4\rho(s_4)}{EI(s_2)}\dd s_4 \dd s_3 \dd s_2 \dd s_1 \nonumber\\
 &\hspace{5mm}+J\int_0^x \int_0^{s_1}\frac{1}{EI(s_2)}\dd s_2 \dd s_1, \label{eq:h1sol} \\[0.5ex]
 &\hspace{-3mm} h_{1,x}(x) = -\int_0^x \!\int_0^{s_1}\!\int_0^{s_2}\! \frac{s_3\rho(s_3)}{EI(s_1)} \dd s_3 \dd s_2 \dd s_1 \nonumber\\
 &\hspace{5mm}+J\int_0^x \frac{1}{EI(s_1)} \dd s_1 \label{eq:h1xsol} \\[-4.5ex]\nonumber
\end{align}
for each $x\in[0,L]$. Since $EI$ and $\rho$ are strictly positive functions belonging to $C^4[0,L]$, it follows from the above equations via successive integrations that $h_1 \in C^4[0,L]$ and the estimates in \eqref{eq:esthhx} hold for $k=1$ with \vspace{-2mm}
\begin{equation}\label{eq:R2defn}
  R_2 = \Bigg(\frac{\max_{x\in [0,L]}\rho(x)+J}{\min_{x\in [0,L]} EI(x)} \Bigg) \max\{1,L^3\}.  \vspace{-1.5mm}
\end{equation}
Solving \eqref{eq:hkODE}-\eqref{eq:hkIC2} for $h_k$ we get that for $k\geq2$, \vspace{-1mm}
\begin{align}
 &h_k(x) = -\int_{0}^{x}\!\int_{0}^{s_1}\!\!\int_{0}^{s_2}\!\! \int_{0}^{s_3} \!\!\frac{\rho(s_4)h_{k-1}(s_4)}{EI(s_2)}\dd s_4 \dd s_3 \dd s_2 \dd s_1, \label{eq:hksol} \\
 &h_{k,x}(x) = -\int_{0}^{x}\!\int_{0}^{s_1}\!\!\int_{0}^{s_2}\!\! \frac{\rho(s_3)h_{k-1}(s_3)}{EI(s_1)} \dd s_3 \dd s_2 \dd s_1. \nonumber\\[-4.5ex]\nonumber
\end{align}
The claim that $h_k\in C^4[0,L]$ and the estimates in \eqref{eq:esthhx} hold for all $k\geq1$ with $R_2$ given in \eqref{eq:R2defn} can be established by mimicking the induction argument given below \eqref{eq:gksol}-\eqref{eq:gkxsol}. \vspace{-4mm}\hfill $\blacksquare$
\end{pf}

The next result shows that, if $y_1$ and $y_2$ belonging to $G_s[0,T]$ are chosen suitably, then $z\in C([0,T];Z)$ determined by $w$ in \eqref{eq:formalsoln} via \eqref{eq:wdetz} is a solution of \eqref{eq:beam1}-\eqref{eq:beam4}.  

\begin{proposition}\label{pr:wisclasic}
Fix $T>0$ and $s\in (1,2)$. Suppose that the functions $y_1, y_2\in G_s[0,T]$ satisfy the following infinite order differential equation: for $t\in[0,T]$, \vspace{-1mm}
\begin{equation}\label{eq:2ndinput}
 \sum_{k\geq 0} g_{k,x}(L)y_1^{(2k)}(t) + \sum_{k\geq 0} h_{k,x}(L)y_2^{(2k)}(t) = 0. \vspace{-2mm}
\end{equation}
Then $w$ given in \eqref{eq:formalsoln} belongs to $C^\infty([0,T];C^4[0,L])$, so that $w(0,\cdot)$, $w(L,\cdot)$  and $w_x(0,\cdot)$ belong to $C^\infty[0,T]$, and $w$ satisfies \eqref{eq:beam1}-\eqref{eq:beam4} with $f(t)=w(L,t)$. Furthermore, the function $z\in C([0,T];Z)$ determined by $w$ via \eqref{eq:wdetz} is the unique solution of \eqref{eq:beam1}-\eqref{eq:beam4} for the initial state $z_0=z(0)$ and control input $f(t)=w(L,t)$. \vspace{-4mm}
\end{proposition}

\begin{pf}
Recall that $EI$ and $\rho$ are strictly positive functions belonging to $C^4[0,L]$. Differentiating the integral expression for $g_k$ in \eqref{eq:gksol} as many times as required, and then bounding the integrand using $\max_{x \in [0,L]} \rho(x)$, $\min_{x \in [0,L]} EI(x)$, $\|EI\|_{C^4[0,L]}$ and the estimate for $g_{k-1}$ from \eqref{eq:estggx}, it follows via successive integrations  that the functions $g_k$, $g_{k,x}$, $g_{k,xx}$, $g_{k,xxx}$ and $g_{k,xxxx}$ are uniformly bounded on $[0,L]$ by $C_g^k/(4k-5)!$. Here $C_g>0$ is a constant independent of $k$. Similarly using the expression for $h_k$ in \eqref{eq:hksol} and the estimate for $h_{k-1}$ in \eqref{eq:esthhx}, we get that the functions $h_k$, $h_{k,x}$, $h_{k,xx}$, $h_{k,xxx}$ and $h_{k,xxxx}$ are uniformly bounded on $[0,L]$ by $C_h^k/(4k-6)!$, where $C_h>0$ is a constant independent of $k$. Since $y_1,y_2\in G_s[0,T]$ for some $s\in (1,2)$, it follows that
$y_1^{(k)}$ and $y_2^{(k)}$ are uniformly bounded on $[0,T]$ by $C_y^{k+1}(k!)^{s}$, where $C_y>0$ is a constant independent of $k$. Using the above uniform bounds, we can bound the terms $g_k(x)y_1^{(2k)}(t)$ and $h_k(x)y_2^{(2k)}(t)$ in the series in \eqref{eq:formalsoln} uniformly on $[0,L]\times [0,T]$ by \vspace{-2mm}
$$C_k = \max\{C_g,C_h\}^k C_y^{2k+1} \frac{(2k!)^s}{(4k-6)!}. \vspace{-2mm}$$
It can be verified using the ratio test that $\sum_{k\geq 0} C_k <\infty$. We can therefore appeal to the Weierstrass M-test to conclude that the series in \eqref{eq:formalsoln} converges uniformly on $[0,L]\times [0,T]$ to a continuous function $w$. We can similarly conclude that every series obtained by termwise differentiation of the series for $w$ in \eqref{eq:formalsoln}, with respect to $x$ (up to four times) and $t$ (any number of times), converges uniformly on $[0,L]\times[0,T]$ to a continuous function. This implies that the derivatives of $w$ with respect to $x$ (up to four times) and $t$ (any number of times) are given by the series obtained by termwise differentiation of the series for $w$ in \eqref{eq:formalsoln} and $w$ has the regularity mentioned in the statement of this proposition.

Using \eqref{eq:g0ODE}, \eqref{eq:g1ODE}, \eqref{eq:gkODE} and \eqref{eq:h0ODE}, \eqref{eq:h1ODE}, \eqref{eq:hkODE} it is easy to check that the series corresponding to $w_{tt}$ and $-(EI w_{xx})_{xx}$ are the same and hence $w$ satisfies the PDE \eqref{eq:beam1}. Taking $x=0$ in the series for $w_{tt}$, $(EI w_{xx})_x$, $w_{xtt}$ and $EIw_{xx}$ and using the initial conditions for $g_k$ in \eqref{eq:g0IC1}, \eqref{eq:g0IC2}, \eqref{eq:g1IC1}, \eqref{eq:g1IC2}, \eqref{eq:gkIC1}, \eqref{eq:gkIC2} and the initial conditions for $h_k$ in \eqref{eq:h0IC1}, \eqref{eq:h0IC2}, \eqref{eq:h1IC1}, \eqref{eq:h1IC2}, \eqref{eq:hkIC1}, \eqref{eq:hkIC2} it follows that $w$ satisfies the boundary conditions \eqref{eq:beam2}-\eqref{eq:beam3}. Finally taking  $f(t)=w(L,t)$ and using \eqref{eq:2ndinput} (which means $w_x(L,t)=0$) we get that $w$ satisfies the boundary condition \eqref{eq:beam4}. In summary, $w$ satisfies \eqref{eq:beam1}-\eqref{eq:beam4} and has the regularity required in Remark \ref{rm:classical}. Hence it follows from the remark that $z$ determined by $w$ via \eqref{eq:wdetz} is the unique solution of \eqref{eq:beam1}-\eqref{eq:beam4} for the initial state $z_0=z(0)$ and control input $f(t)$.  $\hfill \blacksquare$
\end{pf}


\section{Motion planning  \vspace{0mm}} \label{sec4} 

In this section, we present our results for the motion planning problem  of  transferring the beam model \eqref{eq:beam1}-\eqref{eq:beam4} from an initial state $z_0$ to a final state $z_T$ over a prescribed time-interval $[0,T]$ using an appropriate control input $f$. We prove that the desired transfer is possible if the initial and final states belong to a certain set, which contains the steady-states of the beam model and is independent of $T$, see Theorem \ref{th:main_result} and  Remark \ref{rm:ss}.


In the following proposition, we first describe an approach for constructing functions $y_1$ and $y_2$ that satisfy \eqref{eq:2ndinput}. Below we will need the estimate \vspace{-1mm}
\begin{equation} \label{eq:binom}
 (p+q)! \leq 2^{p+q} p! q! \FORALL p,q\in\nline. \vspace{-1mm}
\end{equation}
This estimate follows from the fact that the $(p+1)^{\rm th}$-term (given by the expression $(p+q)!/(p! q!)$) in the binomial expansion of $(1+1)^{p+q}$ is less than $(1+1)^{p+q}$.

\begin{proposition}\label{pr:psolution}
Fix $s\in(1,2)$. Let the operators $\Lscr_1$ and $\Lscr_2$ be defined on $G_s[0,T]$ as follows: For $y \in G_s[0,T]$,
$$\Lscr_1 y =\sum_{k\geq 0} g_{k,x}(L) y^{(2k)}, \qquad \Lscr_2 y =\sum_{k\geq 0} h_{k,x}(L) y^{(2k)}. $$
Then $y_1=\Lscr_2 y$ and $y_2=-\Lscr_1 y$ belong to $G_s[0,T]$ for each $y \in G_s[0,T]$ and satisfy \eqref{eq:2ndinput}. \vspace{-1mm}
\end{proposition}

\begin{pf}
Let $y\in G_s[0,T]$. Then \vspace{-1.5mm}
\begin{equation}\label{eq:pest}
 \sup_{t\in [0,T]}|y^{(k)}(t)|\leq D^{k+1} (k\m!)^s \FORALL k\geq 0 \vspace{-1.5mm}
\end{equation}
and some constant $D>0$. For each $n\geq0$, consider the function series \vspace{-1mm}
\begin{equation} \label{eq:L1derseries}
 \sum_{k\geq 0} h_{k,x}(L) y^{(2k+n)}(t) \quad \textrm{with} \quad t\in[0,T], \vspace{-1.5mm}
\end{equation}
obtained by termwise $n$-times differentiation of the series for $\Lscr_2 y$. Using the estimate for $h_{k,x}$ in \eqref{eq:esthhx} and the estimate in \eqref{eq:pest}, we can bound the $(k+1)^{\rm th}$ term $h_{k,x}(L) y^{(2k+n)}(t)$ in the above series by \vspace{-2mm}
$$D_{k,n} = R_2^{k} L^{4k-3} D^{2k+n+1} \frac{((2k+n)!)^s}{(4k-3)!}, \vspace{-2mm}$$
uniformly on $[0,T]$. Since $s<2$, it follows via the ratio test that $\sum_{k \geq 0} D_{k,n} < \infty$ for each $n\geq0$. We can therefore conclude by applying the Weierstrass M-test that the series in \eqref{eq:L1derseries} converges uniformly on $[0,T]$ to a continuous function for each $n\geq0$. This implies that $y_1=\Lscr_2 y$ is in $C^\infty[0,T]$ and $y_1^{(n)}(t)$ is given by the series in \eqref{eq:L1derseries}, i.e. $y_1^{(n)} = \Lscr_2 y^{(n)}$. From the above discussion, we have $\sup_{t\in [0,T]}|y_1^{(n)}(t)|\leq \sum_{k\geq 0} D_{k,n}$. Using \eqref{eq:binom} with $p=2k$ and $q=n$ to bound $(2k+n)!$ in the expression for $D_{k,n}$ we get after a simple calculation that $D_{k,n} \leq 2^{ns} (n!)^s D^{n} C_0^{k+1} ((2k)!)^s\big/{(4k-3)!}$ for some $C_0>0$ independent of $n$ and $k$. Combining the above estimates we get \vspace{-2mm}
$$\sup_{t\in [0,T]}|y_1^{(n)}(t)|\leq 2^{ns}D^n (n!)^s \sum_{k\geq 0} C_0^{k+1} \frac{((2k)!)^s}{(4k-3)!}. \vspace{-2mm}$$
Since $s<2$, applying the ratio test it follows that the series in the above inequality is convergent and therefore $\sup_{t\in [0,T]}|y_1^{(n)}(t)|\leq C^{n+1} (n!)^s$ for some $C>0$ and all $n\geq0$, i.e. $y_1\in G_s[0,T]$. We can similarly show that $y_2^{(n)} = -\Lscr_1 y^{(n)}$ and $y_2\in G_s[0,T]$.

We will now show that $y_1$ and $y_2$ satisfy \eqref{eq:2ndinput}. Note that for each $y\in G_s[0,T]$, we have \vspace{-1mm}
$$\Lscr_1 \Lscr_2 y = \sum_{k\geq 0}\sum_{j\geq 0} g_{k,x}(L) h_{j,x}(L) y^{(2k+2j)}, \vspace{-2mm}$$
$$\Lscr_2 \Lscr_1 y = \sum_{j\geq 0}\sum_{k\geq 0} g_{k,x}(L) h_{j,x}(L) y^{(2k+2j)}.\vspace{1mm}$$
So $\Lscr_1 \Lscr_2 y $ and $\Lscr_2 \Lscr_1 y$ are both double sums which only differ in the order of summation. Consider the following double sum obtained by rearranging the terms in either $\Lscr_1 \Lscr_2 y$ or $\Lscr_2 \Lscr_1 y$: \vspace{-1.5mm}
\begin{equation}\label{eq:doublesum}
 \sum_{l \geq 0} \sum_{j+k=l} g_{j,x}(L)h_{k,x}(L) y^{(2l)}(t). \vspace{-1.5mm}
\end{equation}
We claim that this double sum is absolutely convergent for each $t\in[0,T]$. Indeed, using \eqref{eq:estggx}, \eqref{eq:esthhx} and \eqref{eq:pest} we get \vspace{-2mm}
\begin{align*}
 &\sum_{l \geq 0} \sum_{j+k=l}| g_{j,x}(L)h_{k,x}(L) y^{(2l)}(t)| \\
 \leq & \sum_{l \geq 0} \sum_{j+k=l} \frac{R_1^j R_2^k L^{4j-2} L^{4k-3} D^{2l+1} ((2l)!)^s}{(4j-2)!(4k-3)!}\\
 \leq & \sum_{l\geq 0}\sum_{j+k=l}\frac{R^{l+1} ((2l)!)^s}{(4l-5)!} = \sum_{l\geq 0}\frac{R^{l+1} ((2l)!)^s (l+1)}{(4l-5)!} <\infty.
\end{align*}
Here $R>0$ is some constant, the second inequality is derived by using the estimate in \eqref{eq:binom} with $p=4j-2$ and $q=4k-3$ and the last inequality is obtained by applying the ratio test. Therefore from Fubini's theorem it follows that the sum of the terms in \eqref{eq:doublesum} is independent of the order of summation, i.e. $\Lscr_1 \Lscr_2 y=\Lscr_2 \Lscr_1 y$. So $y_1=\Lscr_2 y$ and $y_2=-\Lscr_1 y$ satisfy $\Lscr_1 y_1 + \Lscr_2 y_2=0$ or equivalently \eqref{eq:2ndinput}. This completes the proof of the proposition. \hfill $\blacksquare$
\end{pf}

\begin{remark} \label{rm:summary}
We can rewrite \eqref{eq:formalsoln} concisely as \vspace{-1mm}
\begin{equation} \label{eq:formalsolnRE}
 w(x,t) = [\Wscr_1 y_1](x,t) + [\Wscr_2 y_2](x,t) \vspace{-1mm}
\end{equation}
for $x\in[0,L]$ and $t\in[0,T]$. Here the operators $\Wscr_1$ and $\Wscr_2$ are as follows: For $y\in G_s[0,T]$ with $s\in(1,2)$, \vspace{-2mm}
\begin{align*}
 &[\Wscr_1 y](x,t) = \sum_{k\geq 0} g_k(x) y^{(2k)}(t), \\
 &[\Wscr_2 y](x,t) = \sum_{k\geq 0} h_k(x) y^{(2k)}(t).\\[-5ex]
\end{align*}
For any $y\in G_s[0,T]$ with $s\in(1,2)$, it follows from Proposition \ref{pr:psolution} that $y_1 = \Lscr_2 y$ and $y_2 = -\Lscr_1 y$ are in $G_s[0,T]$ and satisfy \eqref{eq:2ndinput}. Letting $y_1 = \Lscr_2 y$ and $y_2 = -\Lscr_1 y$ in \eqref{eq:formalsolnRE} and appealing to Proposition \ref{pr:wisclasic} we can conclude that $w = \Wscr_1 \Lscr_2 y - \Wscr_2\Lscr_1 y$ is in $C^\infty([0,T];C^4[0,L])$. Moreover, we get that the function $z\in C([0,T];Z)$ determined by $w$ via the expression $z(t)=[w(\cdot,t) \ w_t(\cdot,t) \ w_t(0,t) \ w_{xt}(0,t)]$ for all $t\in[0,T]$ is the unique solution of \eqref{eq:beam1}-\eqref{eq:beam4} for the initial state $z_0 = z(0)$ and input $f(t)=w(L,t)$. \hfill$\square$
\end{remark}

Recall the operators $\Lscr_1$, $\Lscr_2$, $\Wscr_1$, $\Wscr_2$ from Proposition \ref{pr:psolution} and Remark \ref{rm:summary}. In the next theorem, building on the results in Propositions \ref{pr:genfnest}, \ref{pr:wisclasic} and \ref{pr:psolution}, we prove by construction the existence of a control input $f$ which transfers the beam model \eqref{eq:beam1}-\eqref{eq:beam4} between any two states belonging to a certain subspace of $Z$ over a prescribed time-interval. \vspace{1mm}

\begin{theorem}\label{th:main_result}
Fix a time $T>0$. Consider the set \vspace{-1mm}
\begin{align*}
 &\quad \Phi = \Big\{[v(\cdot,0) \ v_t(\cdot,0) \ v_t(0,0) \ v_{xt}(0,0)] \in Z \m\Big |\m    \nonumber\\
  & v=\Wscr_1 \Lscr_2 y - \Wscr_2 \Lscr_1 y  \ \textrm{for some} \ y\in G_s[0,T], \ s\in (1,2) \Big\}. \vspace{-1mm}
\end{align*}
Let $z_0\in \Phi$ and $z_T \in \Phi$ be given. Then there exists an $f\in C[0,T]$ compatible for $z_0$ and a unique solution $z\in C([0,T];Z)$ of \eqref{eq:beam1}-\eqref{eq:beam4} on the time interval $[0,T]$ for the initial state $z_0$ and input $f$ such that $z(T)=z_T$. \vspace{-3mm}
\end{theorem}

\begin{pf}
From Remark \ref{rm:summary} it is evident that $\Phi$ is a well-defined and non-empty set, and by definition $\Phi$ is independent of $T$. Since $z_0,z_T\in \Phi$ by assumption and $G_{s_1}[0,T] \subset G_{s_2}[0,T]$ for $s_1\leq s_2$, it follows from the definition of $\Phi$ that there exist $y_0, y_T\in G_s[0,T]$ for some $s\in (1,2)$ such that by defining $v_0 = \Wscr_1 \Lscr_2 y_0 - \Wscr_2\Lscr_1 y_0$ and $v_T = \Wscr_1 \Lscr_2 y_T - \Wscr_2\Lscr_1 y_T$, we can express $z_0$ and $z_T$ as  \vspace{-1mm}
\begin{align}
&z_0 = [v_0(\cdot,0) \ v_{0,t}(\cdot,0) \ v_{0,t}(0,0) \ v_{0,tx}(0,0)], \label{eq:z0exp}\\[1ex]
&z_T = [v_T(\cdot,0) \ v_{T,t}(\cdot,0) \ v_{T,t}(0,0) \ v_{T,tx}(0,0)]. \label{eq:zTexp}
\end{align}
For $t\in[0,T]$ let \vspace{-1mm}
\begin{equation}\label{eq:psi}
 \psi(t)=1-\Big(\int_0^t \psi_0(\tau)\dd \tau\bigg / \int_0^T \psi_0(\tau)\dd \tau\Big), \vspace{-1mm}
\end{equation}
where $\psi_0(t)=\exp\left(-\left[\left(1-\frac{t}{T} \right)\frac{t}{T} \right]^{-\frac{1}{s-1}} \right)$ for $t\in (0,T)$ and $\psi_0(0)=\psi_0(T)=0$. Then $\psi \in G_s[0,T]$ and \vspace{-1mm}
\begin{equation}\label{eq:psival}
 \psi(0) = 1, \quad \psi(T)=0, \quad  \psi^{(k)}(0)=\psi^{(k)}(T)=0 \vspace{-1mm}
\end{equation}
for all $k\geq 1$, see \cite{MeScKu:2010}. For all $t\in [0,T]$ define \vspace{-1mm}
\begin{equation}\label{eq:pdef}
 y(t) = y_0(t)\psi(t)+y_T(T-t)\psi(T-t). \vspace{-1mm}
\end{equation}
Since $y_0, y_T,\psi \in G_s[0,T]$ and $G_s[0,T]$ is closed under addition and multiplication of functions \cite[Proposition 1.4.5]{Ro:1993} we get that $y\in G_s[0,T]$. Let $w = \Wscr_1 \Lscr_2 y -\Wscr_2 \Lscr_1 y$. Then $z(t) = [w(\cdot,t) \ w_t(\cdot,t) \  w_{xt}(\cdot,t) \ w_{xt}(\cdot,t)]\in C([0,T];Z)$ is the unique solution of \eqref{eq:beam1}-\eqref{eq:beam4} for the initial state $z(0)$ and input $f(t)=w(L,t)$, see Remark \ref{rm:summary}. We will now complete the proof of this theorem by establishing that $z(0)=z_0$ and $z(T)=z_T$.

The expression $w = \Wscr_1 \Lscr_2 y - \Wscr_2 \Lscr_1 y$ is the same as \eqref{eq:formalsoln} with $y_1 = \Lscr_2 y$ and $y_2 = - \Lscr_1 y$. The series for $w(\cdot,t)$, $w_t(\cdot,t)$, $w_t(0,t)$ and $w_{xt}(0,t)$ can be obtained by termwise differentiation of the series in \eqref{eq:formalsoln} (see the proof of Proposition \ref{pr:wisclasic}). Therefore $z(t)= [w(\cdot,t) \ w_t(\cdot,t)\ w_t(0,t) \ w_{xt}(0,t)]$ for $t\in[0,T]$ can be rewritten as \vspace{-1mm}
$$ z(t) = \sum_{k\geq0} A_k y_1^{(k)}(t) + \sum_{k\geq0} B_k y_2^{(k)}(t)\quad \forall t\in[0,T], \vspace{-2mm}$$
where $A_k, B_k\in C^4[0,L]\times C^4[0,L]\times\rline\times\rline$. Note that $y_1^{(k)}(t) = [\Lscr_2 y^{(k)}](t)$ and $y_2^{(k)}(t) = -[\Lscr_1 y^{(k)}](t)$ (see the proof of Proposition \ref{pr:psolution}). Using this and the definition of the operators $\Lscr_1$ and $\Lscr_2$ in the above equation it follows that for each $t\in[0,T]$, \vspace{-1mm}
\begin{equation} \label{eq:zdsum}
 \!z(t)\!=\!\sum_{k\geq0} \sum_{n\geq0} \bigl(A_k h_{n,x}(L)-B_k g_{n,x}(L)\bigr) y^{(n+k)}(t). \vspace{-2mm}
\end{equation}
Applying the above procedure used to derive \eqref{eq:zdsum} starting from $w = \Wscr_1 \Lscr_2 y - \Wscr_2 \Lscr_1 y$, to both $v_0 = \Wscr_1 \Lscr_2 y_0 - \Wscr_2 \Lscr_1 y_0$ and $v_T = \Wscr_1 \Lscr_2 y_T - \Wscr_2 \Lscr_1 y_T$ (instead of $w$) and then recalling \eqref{eq:z0exp}-\eqref{eq:zTexp} we get \vspace{-1mm}
$$z_0=\sum_{k\geq0}  \sum_{n\geq0} \bigl( A_k h_{n,x}(L) - B_k g_{n,x}(L)\bigr) y_0^{(n+k)}(0), \vspace{-2mm}$$
$$z_T=\sum_{k\geq0}  \sum_{n\geq0} \bigl( A_k h_{n,x}(L) - B_k g_{n,x}(L)\bigr) y_T^{(n+k)}(0).\vspace{1mm}$$
Differentiating \eqref{eq:pdef} $n$-times and then using \eqref{eq:psival} we get $y^{(n)}(0) = y_0^{(n)}(0)$ and $y^{(n)}(T) = y_T^{(n)}(0)$ for all $n\geq 0$. It now follows from the expressions for $z(0)$ and $z(T)$ from \eqref{eq:zdsum} and the expressions for $z_0$ and $z_T$ given above that $z(0)=z_0$ and $z(T)=z_T$. This completes the proof. \vspace{-2mm} \hfill $\blacksquare$
\end{pf}

When $z_0$ and $z_T$ in Theorem \ref{th:main_result} are real-valued, it is easy to see that the control input $f$ proposed in the proof of the theorem and the corresponding state trajectory of the beam model \eqref{eq:beam1}-\eqref{eq:beam4} are both real-valued. Indeed, the results in Sections \ref{sec2}, \ref{sec3} and \ref{sec4} will remain the same if we take the state-space $Z$ to be real-valued, which is physically more meaningful. However, we have taken $Z$ to be a complex-valued Hilbert space since the eigenfunctions of the beam model are complex-valued and we require them to belong to the state-space, see Section \ref{sec5}. The next remark is about the real-valued steady-states of the beam model.

\begin{remark} \label{rm:ss}
Let $c\in\rline$. For the initial state $[u_{ss} \ 0 \ 0 \ 0]\in Z$ with $u_{ss}(x)=c$ for all $x\in[0,L]$ and the constant input $f_{ss}\in C^\infty([0,T];\rline)$ with $f_{ss}(t)=c$ for all $t\in[0,T]$, note that the constant function $z_{ss}(t)=[u_{ss} \ 0 \ 0 \ 0]$ for $t\in [0,T]$ is the solution of \eqref{eq:beam1}-\eqref{eq:beam4}. This follows from Remark \ref{rm:classical}, with $w(\cdot,t)=u_{ss}$ for $t\in[0,T]$ and $f=f_{ss}$. We call $z_{ss}(0)$ the steady-state of \eqref{eq:beam1}-\eqref{eq:beam4} corresponding to the constant input $f_{ss}$. Each such steady-state is a rest configuration of the beam corresponding to some fixed position of the cantilever joint. \vspace{-1mm}

Let $v = \Wscr_1 \Lscr_2 y - \Wscr_2 \Lscr_1 y$, where $y$ is the constant function given as $y(t)=c$ for all $t\in[0,T]$ and some $c\in\rline$. Using the definitions of the operators $\Lscr_1$, $\Lscr_2$, $\Wscr_1$ and $\Wscr_2$ and the function $y$, we get \vspace{-1mm}
\begin{equation} \label{eq:vsteadystate}
  v(x,t) = g_0(x) h_{0,x}(L) y(t) - h_0(x) g_{0,x}(L) y(t) \vspace{-1mm}
\end{equation}
for all $x\in[0,L]$ and $t\in[0,T]$. Since $g_0(x)=1$ and $h_0(x)=x$, see \eqref{eq:g0} and \eqref{eq:h0}, and $y$ is a constant function, it follows from \eqref{eq:vsteadystate} that the initial state \vspace{-1mm}
$$ [v(\cdot,0) \ \ v_t(\cdot,0) \ \ v_t(0,0) \ \ v_{xt}(0,0)]=[u_{ss} \ \ 0 \ \ 0 \ \ 0], \vspace{-1mm}$$
where $u_{ss}(x)=c$ for all $x\in [0,L]$. By definition, the above state is in $\Phi$ and is also a steady-state. In other words, steady-states of \eqref{eq:beam1}-\eqref{eq:beam4} belong to the set $\Phi$. It now follows from Theorem \ref{th:main_result} that we can find an input $f$ to transfer  \eqref{eq:beam1}-\eqref{eq:beam4} between any two steady-states. \hfill$\square$
\end{remark}

\section{Approximate Controllability} \label{sec5}  

Given initial and final states belonging to a certain subspace $\Phi$ of the state-space $Z$ and a $T>0$, we have shown in Theorem \ref{th:main_result} that there exists a control input $f$ which transfers the beam model \eqref{eq:beam1}-\eqref{eq:beam4} between these two states over the time interval $[0,T]$. In this section we prove that the set $\Phi$ is dense in $Z$, see Proposition \ref{pr:Phidense}. Combining this result with the well-posedness result in Proposition \ref{pr:wellposed}, we then conclude in Theorem \ref{th:approx} that \eqref{eq:beam1}-\eqref{eq:beam4} is approximately controllable over all time intervals $[0,T]$.

Recall the Hilbert space $\tilde{Z}$ introduced below \eqref{eq:wbeam4}. Consider the following state operator ${\Ascr}:\tilde Z\to \tilde Z$ associated with the beam model \eqref{eq:beam1}-\eqref{eq:beam4}: The domain of $\Ascr$ is $\Dscr({\Ascr}) = \{[u \ v \ \alpha \ \beta] \in H^4(0,L)\times H^2(0,L) \times \rline \times \rline \m\big|\m  u(L)=u_x(L)=v(L)=v_x(L)=0,\, v(0) = \alpha, \, v_x(0)=\beta\}$ and $ \Ascr \bigl[u \ v \ \alpha  \ \beta\bigr] = [v \ -\!\frac{(EI u_{xx})_{xx}}{\rho} \ -\!\frac{(EI u_{xx})_x(0)}{m} \  \frac{EI(0)u_{xx}(0)}{J} ]$
for all $[u \ v \ \alpha \ \beta]\in \Dscr(\Ascr)$.

The spectral result given below for the state operator $\Ascr$ follows from \cite{Guo:2002}, which considers a beam model similar to \eqref{eq:beam1}-\eqref{eq:beam4}. (In \cite{Guo:2002} the cantilever joint is at $x=0$, the tip-mass is at $x=1$ and the inputs are different.)

\begin{proposition}\label{pr:eigenprop}
  The spectrum of $\Ascr$ consists of countably many eigenvalues $\{\lambda_n\}_{n \in \nline}$ which lie on the imaginary axis. The corresponding eigenfunctions $\{\phi_n\}_{n \in \nline}$ of $\Ascr$, scaled appropriately, form a Riesz basis for the Hilbert space $\tilde Z$ and are of the form \vspace{-1mm}
  \begin{equation}\label{eq:egn-fnc-rep}
    \phi_n = [\,u_n \ \  \lambda_n u_n \ \  \lambda_n u_n(0) \ \ \lambda_n u_{n,x}(0)\,], \vspace{-1mm}
  \end{equation}
  where $u_n\in C^4[0,L]$ is a non-zero solution of the following boundary value ODE: \vspace{-1mm}
  \begin{align}
    &\lambda_n^2 \rho(x) u(x) + \left(EI u_{xx}\right)_{xx}(x) = 0, \label{eq:eigenode1}\\
    &m\lambda_n^2 u(0) + \left(EI u_{xx}\right)_{x}(0) = 0, \label{eq:eigenode2} \\
    &J \lambda_n^2 u_{x}(0) - EI(0) u_{xx}(0)=0, \label{eq:eigenode3}\\
    &u(L) = 0, \qquad u_{x}(L) = 0. \label{eq:eigenode4} \\[-5ex] \nonumber
  \end{align}
\end{proposition}

Recall the generating functions $g_k$ and $h_k$ from Section \ref{sec3}.
In the next proposition, we prove that the eigenfunctions of $\Ascr$ can be expressed in terms of $g_k$ and $h_k$ and that they belong to the set $\Phi$, and finally conclude that the set $\Phi$ is dense in the state-space $Z$.

\begin{proposition} \label{pr:Phidense}
Consider an eigenvalue $\lambda_n$ of $\Ascr$, its corresponding eigenvector $\phi_n$ and the solution $u_n$ to the boundary value ODE \eqref{eq:eigenode1}-\eqref{eq:eigenode4} which determines $\phi_n$ via \eqref{eq:egn-fnc-rep}. Then $u_n$ can be expressed as \vspace{-1mm}
\begin{equation}\label{eq:un_exp}
 \!\!u_n(x) \!=\! u_n(0)\! \sum_{k \geq 0}\! g_k(x) \lambda_n^{2k} + u_{n,x}(0)\! \sum_{k \geq 0}\! h_k(x) \lambda_n^{2k} \vspace{-1.5mm}
\end{equation}
for all $x\in[0,L]$. Recall $g_0$ from \eqref{eq:g0} and define the vector $\phi_0 = [g_0 \ 0 \ 0 \ 0]$ in $Z$. Then the sequence of functions $\{\phi_n\}_{n \in\{0,\nline\}}$ belongs to the set $\Phi$ and forms a Riesz basis for the state-space $Z$. In other words, $\Phi$ is dense in $Z$. \vspace{-3mm}
\end{proposition}

\begin{pf}
We denote the series on the right-side of \eqref{eq:un_exp} by $\hat u_n$, i.e. \vspace{-1.5mm}
\begin{equation}\label{eq:uhat_exp}
 \!\! \hat u_n(x) \!=\! u_n(0)\! \sum_{k \geq 0}\! g_k(x) \lambda_n^{2k} + u_{n,x}(0)\! \sum_{k \geq 0}\! h_k(x) \lambda_n^{2k} \vspace{-2mm}
\end{equation}
for each $x\in[0,L]$. From the proof of Proposition \ref{pr:wisclasic}, we have that there exists a constant $C>0$ independent of $k$ such that the functions $g_k$, $g_{k,x}$, $g_{k,xx}$, $g_{k,xxx}$, $g_{k,xxxx}$,  $h_k$, $h_{k,x}$, $h_{k,xx}$, $h_{k,xxx}$ and $h_{k,xxxx}$ are uniformly bounded on $[0,L]$ by $C^k/(4k-6)!$. Using this we can bound the terms $g_k(x)\lambda_n^{2k}$ and $h_k(x)\lambda_n^{2k}$ in the series in \eqref{eq:uhat_exp} uniformly on $[0,L]$ by $D_k = C^k |\lambda_n|^{2k}/(4k-6)!$. Since $\sum_{k\geq 0} D_k <\infty$, we can appeal to the Weierstrass M-test to conclude that the series in \eqref{eq:uhat_exp} converges uniformly on $[0,L]$ to a continuous function $\hat u_n$. We can similarly conclude that every series obtained by termwise differentiation (up to four times) of the series for $\hat u_n$ in \eqref{eq:uhat_exp}, converges uniformly on $[0,L]$ to a continuous function. In summary, the function $\hat u_n$ in \eqref{eq:uhat_exp} belongs to $C^4[0,L]$ and its derivatives can be obtained by termwise differentiation of the series in \eqref{eq:uhat_exp}.

We will now show that $\hat u_n = u_n$.
Using \eqref{eq:g0ODE}, \eqref{eq:g1ODE}, \eqref{eq:gkODE} and \eqref{eq:h0ODE}, \eqref{eq:h1ODE}, \eqref{eq:hkODE} it is easy to see that the series corresponding to $\lambda_n^2 \rho(x) \hat u_n(x)$ and $-(EI \hat u_{n,xx})_{xx}(x)$ are the same for each $x\in[0,L]$, and hence $\hat u_n$ satisfies \eqref{eq:eigenode1}. Using the series for $\hat u_n(0)$, $\hat u_{n,x}(0)$, $EI \hat u_{n,xx}(0)$ and $(EI \hat u_{n,xx})_x(0)$ and the initial conditions for $g_k$ in \eqref{eq:g0IC1}, \eqref{eq:g0IC2}, \eqref{eq:g1IC1}, \eqref{eq:g1IC2}, \eqref{eq:gkIC1}, \eqref{eq:gkIC2} and the initial conditions for $h_k$ in \eqref{eq:h0IC1}, \eqref{eq:h0IC2}, \eqref{eq:h1IC1}, \eqref{eq:h1IC2}, \eqref{eq:hkIC1}, \eqref{eq:hkIC2}, it is straightforward to verify $\hat u_n$ satisfies the boundary conditions \eqref{eq:eigenode2}-\eqref{eq:eigenode3}. Finally, from the series for $\hat u_n(0)$ and $\hat u_{n,x}(0)$, it follows using \eqref{eq:g0IC1}, \eqref{eq:g1IC1}, \eqref{eq:gkIC1} and \eqref{eq:h0IC1}, \eqref{eq:h1IC1}, \eqref{eq:hkIC1} that \vspace{-1mm}
\begin{equation} \label{eq:ut=u}
  \hat{u}_n(0) = u_n(0), \qquad \hat{u}_{n,x}(0) = u_{n,x}(0). \vspace{-1mm}
\end{equation}
Let $\tilde{u}_n = u_n - \hat{u}_n$. Since $u_n$ and $\hat u_n$ satisfy \eqref{eq:eigenode1}-\eqref{eq:eigenode3} and \eqref{eq:ut=u} holds, it follows that $\tilde u_n$ satisfies \vspace{-1mm}
\begin{align*}
&\lambda_n^2 \rho(x) \tilde{u}_n(x) + \left(EI \tilde{u}_{n,xx}\right)_{xx}(x) = 0 \quad \ \ \forall x \in (0,L), \\
&\tilde{u}_n(0) = 0, \quad \tilde{u}_{n,x}(0) = 0,\\
&\tilde{u}_{n,xx}(0) = 0, \quad (EI \tilde{u}_{n,xx})_x(0) = 0. \\[-4ex]
\end{align*}
Clearly $\tilde u_n=0$ is the only solution to the above ODE and so $u_n = \hat u_n$, i.e. \eqref{eq:un_exp} holds.

Next we will prove that the sequence of eigenfunctions $\{\phi_n\}_{n \in \nline}$ of the operator $\Ascr$ is contained in the set $\Phi$. Consider an eigenvalue $\lambda_n$ and the corresponding eigenfunction $\phi_n = [u_n \ \lambda_n u_n \ \lambda_n u_n(0) \ \lambda_n u_{n,x}(0)]$. Since $u_n$ satisfies \eqref{eq:eigenode4}, we have $u_{n,x}(L) = 0$.  Substituting the series for $u_n$ from \eqref{eq:un_exp} in this equation we get \vspace{-1mm}
\begin{equation}\label{eq:uxL-series}
  u_n(0) \sum_{j \geq 0} g_{j,x}(L) \lambda_n^{2j} + u_{n,x}(0) \sum_{j \geq 0} h_{j,x}(L) \lambda_n^{2j} = 0. \vspace{-2mm}
\end{equation}
First we claim that $\sum_{j \geq 0} g_{j,x}(L) \lambda_n^{2j}>0$. Indeed, it follows from \eqref{eq:g0}, \eqref{eq:g1xsol} and \eqref{eq:gkxsol} that $g_{0,x}(L) = 0$ and $\operatorname{sgn} g_{j,x}(L) = (-1)^j$ for $j \geq 1$. Since $\lambda_n$ lies on the imaginary axis, see Proposition \ref{pr:eigenprop}, $\operatorname{sgn} \lambda_n^{2j} = (-1)^j$ for $j \geq 1$. So each term in $\sum_{j \geq 0} g_{j,x}(L) \lambda^{2j}$ is positive (except the first term which is zero) and our claim follows. Next we claim that $u_{n,x}(0) \neq 0$. Indeed, it follows from our first claim and \eqref{eq:uxL-series} that if $u_{n,x}(0)=0$, then $u_n(0)=0$ which, via \eqref{eq:un_exp}, leads to the contradiction that the function $u_n$ is zero. So $u_{n,x}(0) \neq 0$. Solving \eqref{eq:uxL-series} for $u_n(0)$, we get the expression $u_n(0) = \alpha \sum_{j \geq 0} h_{j,x}(L) \lambda_n^{2j}$, where $\alpha = -u_{n,x}(0)/\sum_{j \geq 0} g_{j,x}(L) \lambda_n^{2j}$. Our claims above imply that $\alpha$ is well-defined and non-zero. Substituting this expression for $u_n(0)$ in \eqref{eq:un_exp} and then simplifying we get \vspace{-2mm}
\begin{align}
  {u_n}(x) &= \alpha \Big(\sum_{k \geq 0} g_k(x) \sum_{j \geq 0} h_{j,x}(L) \lambda_n^{2j+2k} \nonumber\\
 &\hspace{10mm} - \sum_{k \geq 0} h_k(x) \sum_{j \geq 0} g_{j,x}(L) \lambda_n^{2j+2k} \Big) \label{eq:un_long_exp} \\[-4.5ex] \nonumber
\end{align}
for $x\in[0,L]$. Recall $\Lscr_1$ and $\Lscr_2$ from Proposition \ref{pr:psolution} and $\Wscr_1$ and $\Wscr_2$ from Remark \ref{rm:ss}.
Multiplying both sides of \eqref{eq:un_long_exp} by $e^{\lambda_n t}$ and then using the definitions of $\Lscr_1$, $\Lscr_2$, $\Wscr_1$, $\Wscr_2$ to rewrite the resulting equation  we get \vspace{-1mm}
\begin{equation} \label{eq:untov}
 u_n(x) e^{\lambda_n t} = \alpha (\Wscr_1 \Lscr_2 e^{\lambda_n t} - \Wscr_2 \Lscr_1 e^{\lambda_n t}) \vspace{-1mm}
\end{equation}
for all $x \in [0,L]$ and $t \in [0,T]$. Let $v(x,t) = u_n(x) e^{\lambda_n t}$ and $y(t) = \alpha e^{\lambda_n t}$ for $x \in [0,L]$ and $t\in[0,T]$. Then $y \in G_s[0,T]$ for all $s \in(1,2)$ and \eqref{eq:untov} can be rewritten as
$$v=\Wscr_1 \Lscr_2 y - \Wscr_2 \Lscr_1 y.$$
It now follows from the definition of the set $\Phi$ that the vector $[v(\cdot,0) \ v_t(\cdot,0) \ v_t(0,0) \ v_{xt}(0,0)]$, which is exactly the vector $\phi_n= [u_n \ \lambda_n u_n \ \lambda_n u_n(0) \ \lambda_n u_{n,x}(0)]$ (use the definition of $v$), belongs to $\Phi$. This completes the proof of our assertion that $\Phi$ contains $\{\phi_n\}_{n \in \nline}$. \vspace{-1mm}

We will now complete the proof of this theorem by showing that $\Phi$ is dense in $Z$. Note that the vector $\phi_0=[g_0 \ 0 \ 0 \ 0]$ is a steady state for the beam model \eqref{eq:beam1}-\eqref{eq:beam4} and so it belongs to $\Phi$, see Remark \ref{rm:ss}. From the definitions of $\phi_0$ and $\tilde{Z}$, clearly $\phi_0 \notin \tilde{Z}$. Since $\{\phi_n\}_{n\in \nline}$ forms a Riesz basis for $\tilde{Z}$, see Proposition \ref{pr:eigenprop}, and $Z=\tilde{Z} + \{\beta \phi_0\ | \beta \in \cline \}$, it follows from \cite[Proposition 2.5.5]{TuWe:2009} that the sequence of functions $\{\phi_n\}_{n \in \{0 \cup \nline\}}$ forms a Riesz basis for the state-space $Z$. Finally, since $\Phi$ contains $\{\phi_n\}_{n \in \{0 \cup \nline\}}$, the density of $\Phi$ in $Z$ follows. \vspace{-2mm}$\hfill \blacksquare$
\end{pf}

We now present our approximate controllability result for the beam model
\eqref{eq:beam1}-\eqref{eq:beam4}.

\begin{theorem}\label{th:approx}
  The beam model \eqref{eq:beam1}-\eqref{eq:beam4} is approximately controllable over all time intervals, i.e. given a time interval $[0,T]$, a constant $\epsilon>0$ and states $z_0$ and $z_T$ in $Z$, there exists an $f\in C^2[0,T]$ compatible for $z_0$ such that the unique weak solution $z\in C([0,T];Z)$ of \eqref{eq:beam1}-\eqref{eq:beam4} for the initial state $z_0$ and control input $f$ satisfies \vspace{-1mm}
\begin{equation}\label{eq:ac_est}
 \|z(T) - z_T\|_{Z} \leq \epsilon. \vspace{-2mm}
\end{equation}
\end{theorem}

\begin{pf}
Let $T>0$, $\epsilon>0$ and $z_0, z_T \in Z$ be given. Fix $\epsilon_1, \epsilon_2>0$ such that $2M_T \epsilon_1+\epsilon_2\leq \epsilon$. Here $M_T$ is the constant in \eqref{eq:est_wellposed}. Appealing to the density of $\Phi$ in $Z$ established in Proposition \ref{pr:Phidense}, choose states $\tilde{z}_0, \tilde{z}_T \in \Phi$ such that \vspace{-1mm}
\begin{equation}\label{eq:eps_ineq}
  \|z_0 - \tilde{z}_0\|_{Z} \leq \epsilon_1, \qquad \|z_T -  \tilde{z}_T\|_{Z} \leq \epsilon_2.
\end{equation}
Fix $\tilde{f} \in C^2[0,T]$ compatible for $\tilde z_0$ such that the solution $\tilde{z} \in C([0,T];Z)$ of the beam model \eqref{eq:beam1}-\eqref{eq:beam4} for the initial state $\tilde{z}_0$ and control input $\tilde{f}$ satisfies $\tilde{z}(T) = \tilde{z}_T$ (the existence of such an $\tilde f$ is guaranteed by Theorem \ref{th:main_result}). Suppose that $z_0 =[u_0 \ v_0 \ \alpha_0 \ \beta_0]$ and $\tilde{z}_0 =[\tilde{u}_0 \ \tilde{v}_0 \ \tilde{\alpha_0} \ \tilde{\beta_0}]$. Clearly $\tilde f(0) = \tilde u_0(L)$. Define $f \in C^2[0,T]$ as \vspace{-1mm}
\begin{equation}\label{eq:feps}
  f = \tilde{f} + u_0(L) - \tilde{u}_0(L). \vspace{-1mm}
\end{equation}
Then $f(0) = u_0(L)$, i.e. $f$ is compatible for $z_0$. Let $z\in C([0,T];Z)$ be the solution of \eqref{eq:beam1}-\eqref{eq:beam4} for the initial state ${z}_0$ and control input ${f}$. Clearly $f - \tilde{f}$ is compatible for $z_0-\tilde z_0$. It now follows from the linearity of the beam model that $z - \tilde{z}$ is the solution of \eqref{eq:beam1}-\eqref{eq:beam4} for the initial state ${z}_0 - \tilde{z}_0$ and the control input ${f} - \tilde{f}$, and therefore \vspace{-1mm}
$$\|z(T) - \tilde{z}(T) \|_Z \leq M_T \big(\|z_0 - \tilde{z}_0 \|_Z + \| {f} - \tilde{f}\|_{C^2[0,T]}\big), \vspace{-1mm}$$
see \eqref{eq:est_wellposed}. Using $\tilde{z}(T) = \tilde{z}_T$, \eqref{eq:eps_ineq}, \eqref{eq:feps} and the above estimate, it follows via the triangle inequality that
\begin{align*}
  \|z(T) - z_T \|_Z  &\leq \|z(T) - \tilde{z}_T \|_Z + \|z_T - \tilde{z}_T \|_Z  \\[1ex]
  & \leq M_T(\epsilon_1 + | u_0(L) - \tilde{u}_0(L)|) +\epsilon_2 \\[1ex]
  & \leq 2M_T\epsilon_1 +\epsilon_2\  \leq \ \epsilon,
\end{align*}
i.e. \eqref{eq:ac_est} holds. So we have shown that given any $T>0$, $\epsilon>0$ and $z_0, z_T \in Z$, there exists a $f\in C^2[0,T]$ compatible for $z_0$ such that the solution $z\in C([0,T];Z)$ of \eqref{eq:beam1}-\eqref{eq:beam4} for the initial state $z_0$ and input $f$ satisfies \eqref{eq:ac_est}. Hence the beam model \eqref{eq:beam1}-\eqref{eq:beam4} is approximately controllable over all time intervals. $\hfill \blacksquare$
\end{pf}

\section{Numerical and experimental results} \label{sec6} 

Figure 1 shows our experimental setup. It contains a non-uniform moving cantilever beam made of stainless steel. The width of the beam varies linearly along the length of the beam. One end of the beam supports a tip-mass, while the other end is attached via a cantilever joint to a cart mounted on a Hiwin single axis robot. The robot is driven by a Yasaka AC servo motor which fixes the cart position as per the position control input it receives from a NI MyRIO controller. The dynamics of the beam in the setup is governed by the coupled PDE-ODE model \eqref{eq:beam1}-\eqref{eq:beam4} with the following parameters: $L=0.5\m\textrm{m}$, $m=0.4\m\textrm{kg}$, $J=1.86 \times 10^{-4} \m\textrm{kg\m m}^2$ and $\rho(x)=0.11(1+3x)\m\textrm{kg/m}$ and $EI(x)=0.29(1+3x)\m\textrm{N/m}$ for $x\in [0,L]$. In Section \ref{sec6a}, using these parameters, we have numerically validated our solution to the motion planning problem by considering two examples. The initial state of the beam in the first example is a steady-state of the beam, while in the second example it is not a steady-state. The desired final state in both the examples is the zero state. In Section \ref{sec6b}, using our setup, we have experimentally validated the control input obtained in the first example.


\subsection{Numerical validation}\label{sec6a}

We consider two examples to illustrate our solution in Theorem \ref{th:main_result} to the motion planning problem. In both the examples, we consider the beam model \eqref{eq:beam1}-\eqref{eq:beam4} with parameters $L$, $m$, $J$, $\rho$ and $EI$ corresponding to the experimental setup in Figure 1.

\textbf{Example 1}. Consider the vectors $z_0 = [0.3 \ 0 \ 0 \ 0]$ and $z_T = [0 \ 0 \ 0 \ 0]$ in $Z$. These vectors are determined by \eqref{eq:z0exp} and \eqref{eq:zTexp} using $v_0 = \Wscr_1\Lscr_2 y_0 - \Wscr_2 \Lscr_1 y_0$ and $v_T = \Wscr_1\Lscr_2 y_T - \Wscr_2 \Lscr_1 y_T$, where $y_0(t)=0.3$ and $y_T(t)=0$ for all $t\geq0$. These vectors are steady-states of the beam model \eqref{eq:beam1}-\eqref{eq:beam4} and they belong to the set $\Phi$ introduced in Theorem \ref{th:main_result}, see Remark  \ref{rm:ss}. In this example, we address the motion planning problem of finding a control input $f$ which transfers the beam model from the initial state $z_0$ to the final state $z_T$ in 3 seconds.

We solve the above problem by following the procedure described in the proof of Theorem \ref{th:main_result}. Accordingly we choose $\psi$ to be the function in \eqref{eq:psi} with $s=1.5$ and $T=3$ and, using $y_0$ and $y_T$ given above, define $y\in G_{1.5}[0,3]$ via \eqref{eq:pdef}. The required control input is $f(t)=w(L,t)$ for $t\in[0,3]$, where $w=[\Wscr_1 \Lscr_2  - \Wscr_2 \Lscr_1]y$. Using the definitions of $\Lscr_1$, $\Lscr_2$, $\Wscr_1$, $\Wscr_2$ and changing the order of the double summations we obtain the following series representation for the control input $f$: \vspace{-2mm}
$$f(t) = \sum_{l=0}^{\infty}\sum_{j+k=l} \big[g_k(L)h_{j,x}(L) - h_k(L)g_{j,x}(L) \big]y^{(2l)}(t).  \vspace{-2mm}$$
(Note that $g_k$ and $h_k$ are computed using the expressions in the proof of Proposition \ref{pr:genfnest}.) Consider the truncation  \vspace{-2mm}
$$f^N(t) = \sum_{l=0}^{N}\sum_{j+k=l} \big[g_k(L)h_{j,x}(L) -h_k(L)g_{j,x}(L) \big]y^{(2l)}(t) \vspace{-2mm}$$
of the series for $f$. The truncated series $f^N(t)$ converges rapidly to $f(t)$ as $N$ tends to infinity. In this example, $f^{10}$ is a good approximation for the control input $f$ which solves the motion planning problem, see Figure 2. \newpage

\m\vspace{-15mm}
$$\includegraphics{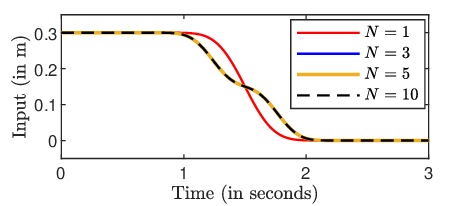}\vspace{-1mm}$$
\centerline{ \parbox{3.2in}{\small
Figure 2. Plot of $f^N$ for $N \in\{1,3,5,10\}$ in Example 1. Here $\|f^5 - f^3\|_{L^2(0,3)} < 9 \times 10^{-7} $ and $\|f^{10} - f^5\|_{L^2(0,3)} < 3 \times 10^{-13} $, which indicate that $f^{10}$ is a good approximation of the control input $f$ which solves the motion planning problem. \vspace{-3mm}}}

$$\includegraphics{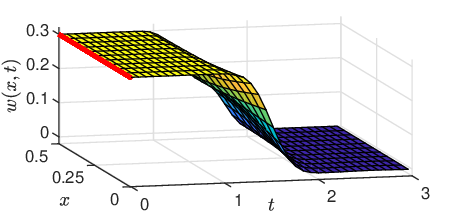}\vspace{-1mm}$$
\centerline{ \parbox{3.2in}{\small
Figure 3. Displacement profile $w(x,t)$ of the beam model \eqref{eq:beam1}-\eqref{eq:beam4} corresponding to the initial state $z_0$ and control input $f^{10}$ in Example 1. The final displacement at $t = 3$ seconds is close to zero ($\|w(\cdot,3)\|_{H^2(0,0.5)} < 4.1 \times 10^{-3}$). \vspace{-3mm}}}

$$\includegraphics{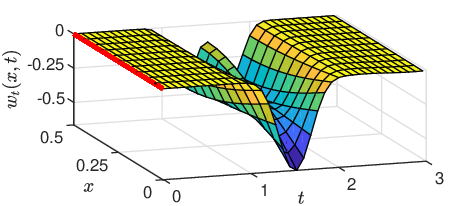}\vspace{-1mm}$$
\centerline{ \parbox{3.2in}{\small
Figure 4. Velocity profile $w_t(x,t)$ of the beam model \eqref{eq:beam1}-\eqref{eq:beam4} corresponding to the initial state $z_0$ and control input $f^{10}$ in Example 1. The final velocity at $t = 3$ seconds is close to zero ($\|w_t(\cdot,3)\|_{L^2(0,0.5)} < 8.1 \times 10^{-4}$).}}

We discretize the spatial derivatives in the beam model \eqref{eq:beam1}-\eqref{eq:beam4} using the finite-difference method (with step-size 1/600) to obtain a set of ODEs in time which serve as a numerical model for the beam model. We validate our solution to the motion planning problem presented above by simulating the numerical model with the initial state $z_0$ and the control input $f^{10}$ shown in Figure 2. Figure 3 and Figure 4 show the beam displacement profile $w(x,t)$ and the beam velocity profile $w_t(x,t)$ obtained from our simulation. Recall that $z_T=[0 \ 0\ 0\ 0]$. As expected from Theorem \ref{th:main_result}, the displacement profile and the velocity profile converge to the zero function at $t=3$ seconds.

\textbf{Example 2}. Let $y_0(t)=1 + 10t^2e^{-2t}$ and $y_T(t)=0$ for all $t\in[0,3]$. Note that $y_0$ is a real analytic function and so $y_0\in G_{1.5}[0,3]$. Let the vectors $z_0, z_T\in Z$ be determined by \eqref{eq:z0exp} and \eqref{eq:zTexp} using $v_0 = \Wscr_1\Lscr_2 y_0 - \Wscr_2 \Lscr_1 y_0$ and $v_T = \Wscr_1\Lscr_2 y_T - \Wscr_2 \Lscr_1 y_T$. Clearly $z_T=[0 \ 0 \ 0 \ 0]$ and $z_0, z_T$ belong to the set $\Phi$ introduced in Theorem \ref{th:main_result}. Furthermore, $v_0(\cdot,0)$ (red line in Figure 6) is not a constant function and $v_{0,t}(\cdot,0)$ (red line in Figure 7) is not the zero function. So $z_0$ is not a steady-state of the beam equation \eqref{eq:beam1}-\eqref{eq:beam4}, see Remark \ref{rm:ss}. In this example, we address the motion planning problem of finding a control input $f$ which transfers the beam model from the initial state $z_0$ to the final state $z_T$ in 3 \vspace{-7.5mm} seconds.

$$\includegraphics{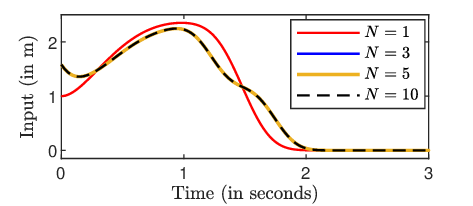}\vspace{-1.5mm}$$
\centerline{ \parbox{3.2in}{\small
Figure 5. Plot of $f^N$ for $N \in\{1,3,5,10\}$ in Example 2. Here $\|f^5 - f^3\|_{L^2(0,3)} < 7 \times 10^{-6} $ and $\|f^{10} - f^5\|_{L^2(0,3)} < 3 \times 10^{-12} $, which indicate that $f^{10}$ is a good approximation of the control input $f$ which solves the motion planning problem. \vspace{-1mm} }}

$$\includegraphics{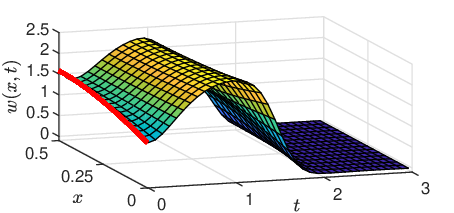}\vspace{-1mm}$$
\centerline{ \parbox{3.2in}{\small
Figure 6. Displacement profile $w(x,t)$ of the beam model \eqref{eq:beam1}-\eqref{eq:beam4} corresponding to the initial state $z_0$ and control input $f^{10}$ in Example 2. The final displacement at $t = 3$ seconds is close to zero ($\|w(\cdot,3)\|_{H^2(0,0.5)} < 2.3 \times 10^{-2}$). \vspace{-5mm}}}

$$\includegraphics{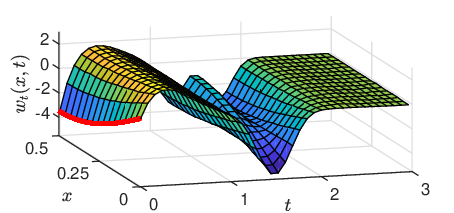}\vspace{-1mm}$$
\centerline{ \parbox{3.2in}{\small
Figure 7. Velocity profile $w_t(x,t)$ of the beam model \eqref{eq:beam1}-\eqref{eq:beam4} corresponding to the initial state $z_0$ and control input $f^{10}$ in Example 2. The final velocity at $t = 3$ seconds is close to zero ($\|w_t(\cdot,3)\|_{L^2(0,0.5)} < 1.6 \times 10^{-2}$). \vspace{-2mm}}}

Like in Example 1, we solve the above problem by following the procedure described in the proof of Theorem \ref{th:main_result}. Accordingly we choose $\psi$ to be the function in \eqref{eq:psi} with $s=1.5$ and $T=3$ and, using $y_0$ and $y_T$ given above, define $y\in G_{1.5}[0,3]$ via \eqref{eq:pdef}. The required control input is $f(t) = w(L,t)$ for $t\in[0,3]$, where $w=[\Wscr_1 \Lscr_2  - \Wscr_2 \Lscr_1]y$. Truncating the series for $f$ like in Example 1 we obtain the finite series $f^N$, which converges rapidly to $f$ as $N$ tends to infinity. In this example, $f^{10}$ is a good approximation for the control input $f$ which solves the motion planning problem, see Figure 5. \vspace{-1mm}

We validate our above solution to the motion planning problem by simulating the finite-difference numerical model for the beam model, described in Example 1, with the non-steady initial state $z_0$ and the control input $f^{10}$ shown in Figure 5. Figure 6 and Figure 7 show the beam displacement profile $w(x,t)$ and the beam velocity profile $w_t(x,t)$ obtained from our simulation. As expected from Theorem \ref{th:main_result}, the displacement profile and the velocity profile converge to the zero function at $t=3$ seconds. \vspace{-9mm}

\subsection{Experimental validation}\label{sec6b} \vspace{-2mm}

In Example 1, we have constructed a smooth control input $f^{10}$ (see Figure 2) for transferring the beam model from $z_0=[0.3\ 0\ 0 \ 0]$ to $z_T=[0 \ 0 \ 0 \ 0]$ in 3 seconds. Note that the control input $f^{10}(t)$ specifies the displacement of the cart (cantilever joint) at the time instant $t$. We implement this control input $f^{10}$ on our experimental setup as follows. We discretize the time interval $[0,3]$ seconds into subintervals of 20\m{ms} each. For each subinterval $[a,b]$, we compute $f^{10}(b)-f^{10}(a)$, which is the required displacement of the cart over the time interval $[a,b]$. We then convert this displacement into the number of pulses, which when provided to the AC servo drive over the time interval $[a,b]$, ensures that the single axis robot moves the cart by $f^{10}(b)-f^{10}(a)$. We then run a 20\m{ms} loop in the FPGA module in the NI MyRIO controller for 3 seconds, which provides the computed number of pulses over each 20\m{ms} subinterval to the servo drive, thereby ensuring that the cart position tracks the control input $f^{10}$ accurately. The position of the tip-mass during this experiment is measured using a Micro-Epsilon laser displacement sensor, which has an accuracy of 0.1\m{mm} and a sampling rate of 2000\m{Hz}. Figure 8 shows the position of the tip-mass obtained from our experiment and the numerical simulation in Example 1. As seen from the plots, they match closely. The video of our experiment showing that the beam is transferred from one rest position to another at a distance of 0.3\m{m} as desired is available here: {https://youtu.be/-iXZiz25g5Q}. In other words, our solution $f^{10}$ to the motion planning problem in Example 1 performs as expected while implemented on the experimental setup.

$$\includegraphics{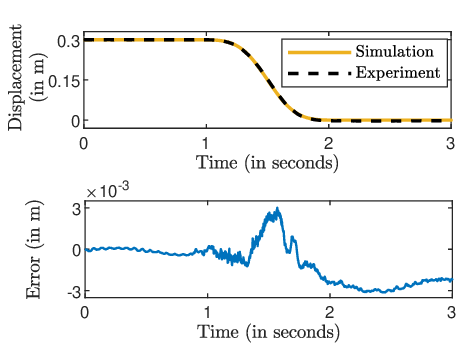}\vspace{-3.5mm}$$
\centerline{ \parbox{3.2in}{\small
Figure 8. The first plot shows the tip-mass position obtained from our experiment and the numerical simulation in Example 1, while the second plot shows the error between these two positions. The maximum error is about 3\m{mm} (i.e. $1$\% of the net dispacement of the tip-mass). \vspace{1mm} }}

\appendix

{\bf Appendix A \vspace{0mm}}
\renewcommand{\theequation}{{A.\arabic{equation}}}

{\bf Proof of the claim above \eqref{eq:weaksol_wbeam}.} \vspace{-2mm}

The coupled PDE-ODE system \eqref{eq:wbeam1}-\eqref{eq:wbeam4} can be written as an abstract evolution equation on $\tilde Z$ as follows: \vspace{-2mm}
\begin{equation}\label{eq:evol}
 \dot{\tilde{z}} (t) = \Ascr \tilde z(t) + \Bscr\bigl[ f(t) \ \ \ddot f(t)\bigr] \qquad \forall \, t>0. \vspace{-2mm}
\end{equation}
Here the state operator $\Ascr:\tilde Z\to \tilde Z$ has domain $\Dscr(\Ascr) = \{[u \ v \ \alpha \ \beta] \in H^4(0,L)\times H^2(0,L) \times \rline \times \rline \m\big|\m  u(L)=u_x(L)=v(L)=v_x(L)=0,\, v(0) = \alpha, \, v_x(0)=\beta\}$ and
$ \Ascr \bigl[u \ v \ \alpha  \ \beta\bigr] = [v \ -\!\frac{(EI u_{xx})_{xx}}{\rho} \ -\!\frac{(EI u_{xx})_x(0)}{m} \  \frac{EI(0)u_{xx}(0)}{J} ]$ for all $[u \ v \ \alpha \ \beta]\in \Dscr(\Ascr)$ and the control operator $\Bscr:\rline^2\to\tilde Z$ is a bounded linear operator such that $\Bscr [ g_1 \ g_2] = [0 \ \frac{(EI  \nu_{xx})_{xx}g_1+\nu g_2}{ \rho} \ \ 0 \ \ 0]$ for all $[g_1 \ g_2]\in \rline^2$. The operator $\Ascr$ generates a $C_0$-semigroup $\tline$ on $\tilde Z$,\vspace{-1mm} see \cite[Section 5]{ZhWe:2011}.

For the initial state $\tilde z_0=[\tilde u_0 \ \tilde v_0 \ \tilde \alpha_0 \ \tilde \beta_0]\in \tilde Z$ and control input $f\in C^2[0,T]$, the mild solution of the abstract system \eqref{eq:evol} is the function $\tilde z \in C([0,T];\tilde Z)$ given by \vspace{-2.5mm}
\begin{equation} \label{eq:tildems}
 \tilde z(t) = \tline_t \tilde z_0 + \int_0^t \tline_{t-\tau}\Bscr \bigl[f(\tau) \ \ \ddot f(\tau)\bigr]\dd \tau. \vspace{-2.5mm}
\end{equation}
From \cite[Remark 4.2.6]{TuWe:2009} it follows that this mild solution  is the unique function in $C([0,T];\tilde Z)$ which satisfies \vspace{-2.5mm}
\begin{align}
 \langle \tilde z(t) -\tilde z_0, p\rangle_{\tilde Z} =& \int_0^t \Big[\langle \tilde z(\tau),-\Ascr p\rangle_{\tilde Z}\nonumber\\
 &\hspace{6mm}+ \big\langle \Bscr \bigl[f(\tau) \ \  \ddot f(\tau)\bigl],p\big\rangle_{\tilde Z}\Big]\dd \tau \label{eq:abs_weakform} \\[-4.5ex]\nonumber
\end{align}
for all $p\in\Dscr(\Ascr)$ and each $t\in [0,T]$. Clearly $\tilde z(t) = [\tilde w(\cdot,t) \ \tilde v(\cdot,t) \ \tilde a(t) \ \tilde b(t)]$, where $\tilde w\in C([0,T];H^2(0,L))$ with $\tilde w(L,t)=\tilde w_x(L,t)=0$, $\tilde v\in C([0,T];L^2(0,L))$ and $\tilde a,\tilde b\in C[0,T]$. For each $t\in[0,T]$, define $\eta_1(t)\in L^2(0,L)$,  $\eta_2(t)\in \rline$ and $\eta_3(t)\in \rline$ as follows: \vspace{-2mm}
\begin{align*}
 &\eta_1(t) = \tilde w(t) - \tilde u_0 - \int_0^t \tilde v(\tau)\dd \tau, \\
 &\eta_2(t) = \tilde w(0,t)- \tilde u_0(0)- \int_0^t \tilde a(\tau) \dd\tau,\\
 &\eta_3(t) = \tilde w_x(0,t)- \tilde u_{0,x}(0)- \int_0^t \tilde b(\tau) \dd\tau. \\[-4.5ex]
\end{align*}
Fix $s\in [0,T]$. Since $\Ascr$ is boundedly invertible, see \cite[Lemma 2.1]{Guo:2002}, there exists a $\tilde p\in \Dscr(\Ascr)$ such that $\Ascr \tilde p = [0 \ \eta_1(s) \ \eta_2(s) \ \eta_3(s)]$. From the definition of $\Ascr$ it follows that $\tilde p$ is of the form $\tilde p=[\psi \ 0 \ 0 \ 0]$ for some $\psi\in H^4(0,L)$ with $\psi(L)=\psi_x(L)=0$. Taking  $p=\tilde p$ in \eqref{eq:abs_weakform}, it follows via a simple calculation using integration by parts that $\eta_1(s)=0$, $\eta_2(s)=0$ and $\eta_3(s)=0$. This holds for all $s\in[0,T]$. So we can conclude that  $\tilde w \in C([0,T];H^2(0,L))\cap C^1([0,T];L^2(0,L))$ with $\tilde w(0,\cdot), \tilde w_x(0,\cdot)\in C^1[0,T]$ and $\tilde w(L,t)=\tilde w_x(L,t)=0$, $[\tilde w(\cdot,0) \ \tilde w_t(\cdot,0) \ \tilde w_t(0,0) \ \tilde w_{xt}(0,0)]=[\tilde u_0 \ \tilde v_0 \ \tilde \alpha_0 \ \tilde \beta_0]$ and \vspace{-1mm}
\begin{equation}\label{eq:tildez}
 \tilde z(t) = [\tilde w(\cdot,t) \ \ \tilde w_t(\cdot,t) \ \ \tilde w_t(0,t) \ \ \tilde w_{xt}(0,t)] \vspace{-1mm}
\end{equation}
for each $t\in[0,T]$.

Any $p\in \Dscr(\Ascr)$ can be expressed  as $[\psi \ \varphi \ \varphi(0) \ \varphi_x(0)]$, where $\psi \in H^4(0,L)$ with $\psi(L)=\psi_x(L)=0$ and $\varphi\in H^2(0,L)$ with $\varphi(L)=\varphi_x(L)=0$. Substituting this expression for $p$ and the expression for $\tilde z$ from \eqref{eq:tildez} into \eqref{eq:abs_weakform} and then collecting the terms which contain only $\varphi$ and $\psi$ separately, it follows that \eqref{eq:abs_weakform} is equivalent to the following pair of equations: \eqref{eq:weaksol_wbeam} and \vspace{-2mm}
\begin{align}
 &\int_0^L EI(x)\bigl[\tilde w_{xx}(x,t) -  \tilde u_{0xx}(x) \bigr]\psi_{xx}(x)\dd x \nonumber\\
 &= \int_0^t\int_{0}^{L}(EI\psi_{xx})_{xx}(x)\tilde w_t(x,\tau)\dd x \dd \tau\nonumber\\
 &\hspace{10mm} - (EI\psi_{xx})_x(0) \int_0^t   \tilde w_t(0,\tau)\dd \tau\nonumber\\
 &\hspace{10mm} + EI(0)\psi_{xx}(0)\int_0^t \tilde w_{xt}(0,\tau)\dd \tau, \label{eq:weaksol_wbeam1} \\[-4.5ex]\nonumber
\end{align}
From the properties of $[\tilde w(\cdot,t) \ \tilde w_t(\cdot,t) \ \tilde w_t(0,t) \ \tilde w_{xt}(0,t)]$ and $\psi$ it is easy to verify via integration by parts that \eqref{eq:weaksol_wbeam1} holds trivially for all $\psi$. In other words, \eqref{eq:abs_weakform} is equivalent to \eqref{eq:weaksol_wbeam}. Since $\tilde z$ in \eqref{eq:tildez} is the unique function in $C([0,T];\tilde Z)$ which satisfies \eqref{eq:abs_weakform} for all $p\in\Dscr(\Ascr)$ and $t\in[0,T]$, it follows that it is also the unique function in $C([0,T];\tilde Z)$ which satisfies \eqref{eq:weaksol_wbeam} for all $\varphi\in H^2(0,L)$ with $\varphi(L)=\varphi_x(L)=0$ and $t\in[0,T]$. This completes the proof of our claim.

Finally, since $\Bscr$ is a bounded linear operator, it follows from \eqref{eq:tildems} and the properties of $C_0$-semigroups that there exists a constant $\tilde M_T>0$ independent of $\tilde z_0$ and $f$ such that \eqref{eq:ztd_wellposed_est} holds.


\end{document}